\documentclass[english,10.5pt]{article}
\usepackage[T1]{fontenc}
\usepackage[latin9]{inputenc}
\usepackage{geometry}
\usepackage{fancyhdr}
\usepackage{mathrsfs}
\usepackage{mathtools}
\usepackage{bm}
\usepackage{algorithm2e}
\usepackage{amsmath}
\usepackage{amsthm}
\usepackage{amssymb}
\usepackage{stmaryrd}
\usepackage{esint}

\makeatletter

\providecommand{\tabularnewline}{\\}

\newcommand{\lyxaddress}[1]{
\par {\raggedright #1
\vspace{1.4em}
\noindent\par}
}
\theoremstyle{plain}
\newtheorem{thm}{\protect\theoremname}
  \theoremstyle{definition}
  \newtheorem{defn}[thm]{\protect\definitionname}
  \theoremstyle{plain}
  \newtheorem{lem}[thm]{\protect\lemmaname}
  \theoremstyle{plain}
  \newtheorem{conjecture}[thm]{\protect\conjecturename}
  \theoremstyle{remark}
  \newtheorem{rem}[thm]{\protect\remarkname}
  \theoremstyle{plain}
  \newtheorem{prop}[thm]{\protect\propositionname}
  \theoremstyle{remark}
  \newtheorem{conclusion}[thm]{\protect\conclusionname}

\usepackage{amsmath}
\usepackage{amsthm}
\usepackage{amssymb}
\usepackage{hyperref}
\usepackage{graphics}
\usepackage{algorithmic}
\usepackage{subfig}
\usepackage{url}
\usepackage{bm}

\usepackage{nccmath}

\usepackage{setspace}

\usepackage{slashed}

\usepackage{esint}
\usepackage{stmaryrd}

\makeatletter

\usepackage{amsfonts}
\usepackage{ascmac}

\usepackage{framed,color}

\definecolor{shadecolor}{gray}{0.9}

\usepackage{ulem}

\makeatother

\allowdisplaybreaks[4]

\makeatother

\usepackage{babel}
  \providecommand{\conclusionname}{Conclusion}
  \providecommand{\conjecturename}{Conjecture}
  \providecommand{\definitionname}{Definition}
  \providecommand{\lemmaname}{Lemma}
  \providecommand{\propositionname}{Proposition}
  \providecommand{\remarkname}{Remark}
\providecommand{\theoremname}{Theorem}

\begin{document}

\title{A Brownian Particle and Fields I:\\
Construction of Kinematics and Dynamics}

\author{{\Large{}Keita Seto}\thanks{keita.seto@eli-np.ro}}

\maketitle

\lyxaddress{\begin{center}
Extreme Light Infrastructure \textendash{} Nuclear Physics (ELI-NP)
/ \\
Horia Hulubei National Institute for R\&D in Physics and Nuclear Engineering
(IFIN-HH), \\
30 Reactorului St., Bucharest-Magurele, jud. Ilfov, P.O.B. MG-6, RO-077125,
Romania.
\par\end{center}}
\begin{abstract}
Tracking a ``real'' trajectory of a quantum particle still has been
treated as the interpretation problem. It shall be expressed by a
Brownian (stochastic) motion suggested by E. Nelson, however, the
well-defined mechanism of field generation from a stochastic particle
hasn't been proposed yet. For the improvement of this, I propose the
extension of Nelson's quantum dynamics, for describing a relativistic
scalar electron with its radiation equivalent to the Klein-Gordon
particle and field system.
\end{abstract}
$\,$$\,$$\,$$\,$

Keyword:

{[}Physics{]} Stochastic quantum dynamics, relativistic motion, field
generation

{[}mathematics{]} Applications of stochastic analysis

$\,$$\,$$\,$$\,$

\maketitle

\thispagestyle{fancy}
\rhead{\texttt{preprint:ELI-NP/RA5-TDR 0003}} 
\lhead{}
\renewcommand{\headrulewidth}{0.0pt}

\setstretch{1.2}

\vfill{}
\newpage{}\thispagestyle{fancy}
\rhead{Keita Seto}
\lhead{ELI-NP/IFIN-HH}
\renewcommand{\headrulewidth}{1.0pt}\tableofcontents

$\,$$\,$$\,$\pagebreak{}

\begin{center}
{\bf Notation and Conventions}
\par\end{center}

\begin{center}
\begin{tabular}{ll}
\hline 
Symbol & Description\tabularnewline
\hline 
\hline 
$c$ & Speed of light\tabularnewline
$\hbar$ & Planck's constant\tabularnewline
$m_{0}$ & Rest mass of an electron\tabularnewline
$e$ & Charge of an electron\tabularnewline
$\mathbb{V_{\mathrm{M}}^{\mathrm{4}}}$ & 4-dimensional standard vector space for the metric affine space\tabularnewline
$g$ & Metric on $\mathbb{V_{\mathrm{M}}^{\mathrm{4}}}$; $g\coloneqq\mathrm{diag}(+1,-1,-1,-1)$\tabularnewline
$\mathbb{A}^{4}(\mathbb{V_{\mathrm{M}}^{\mathrm{4}}},g)$ & 4-dimensional metric affine space with respect to $\mathbb{V_{\mathrm{M}}^{\mathrm{4}}}$
and $g$\tabularnewline
$\mathscr{B}(I)$ & Borel $\sigma$-algebra of a topological space $I$\tabularnewline
$(\mathbb{A}^{4}(\mathbb{V_{\mathrm{M}}^{\mathrm{4}}},g),\mathscr{B}(\mathbb{A}^{4}(\mathbb{V_{\mathrm{M}}^{\mathrm{4}}},g)),\mu)$ & Minkowski spacetime\tabularnewline
$\varphi_{E}\coloneqq\{\varphi_{E}^{A}|A\in\mathrm{set\,of\,indexes}\}$ & Coordinate mapping on $E$; $\varphi_{E}^{A}:E\rightarrow\mathbb{R}$\tabularnewline
$(\mathit{\Omega},\mathcal{F},\mathscr{P})$ & Probability space\tabularnewline
$\mathbb{E}\llbracket\hat{X}(\bullet)\rrbracket\coloneqq\int_{\Omega}d\mathscr{P}(\omega)\,\hat{X}(\omega)$ & Expectation of $\hat{X}(\bullet)\coloneqq\{\hat{X}(\omega)|\omega\in\varOmega\}$\tabularnewline
$\mathbb{E}\llbracket\hat{X}(\bullet)|\mathcal{\mathscr{C}}\rrbracket$ & Conditional expectation of $\hat{X}(\bullet)$ given $\mathcal{\mathscr{C}}\subset\mathcal{F}$\tabularnewline
$\mathcal{\mathscr{P}}_{\tau}\subset\mathcal{F}$  & Sub-$\sigma$-algebra in the family of ''the past'', $\{\mathcal{\mathscr{P}}_{\tau}\}_{\tau\in\mathbb{R}}$\tabularnewline
$\mathscr{F}_{\tau}\subset\mathcal{F}$  & Sub-$\sigma$-algebra in the family of ''the future'', $\{\mathscr{F}_{\tau}\}_{\tau\in\mathbb{R}}$\tabularnewline
$\hat{x}(\circ,\bullet)$ & Dual progressively measurable process (D-progressive, D-process) \tabularnewline
 & {[}Relativistic kinematics of a scalar (spin-less) electron{]}\tabularnewline
 & $\hat{x}(\circ,\bullet)\coloneqq\{\hat{x}(\tau,\omega)\in\mathbb{A}^{4}(\mathbb{V_{\mathrm{M}}^{\mathrm{4}}},g))|\tau\in\mathbb{R},\omega\in\varOmega\}$\tabularnewline
$\mathscr{\mathcal{V}}\in\mathbb{V_{\mathrm{M}}^{\mathrm{4}}}\oplus i\mathbb{V_{\mathrm{M}}^{\mathrm{4}}}$ & Complex velocity; $\mathcal{V}^{\alpha}(x)\coloneqq i\lambda^{2}\times\partial{}^{\alpha}\ln\phi(x)+e/m_{0}\times A{}^{\alpha}(x)$\tabularnewline
\hline 
\end{tabular}
\par\end{center}

\pagebreak{}

\section{Introduction}

This series of the papers proposes quantum dynamics coupled with stochastic
kinematics of a scalar (spin-less) electron and its radiation mechanism
as the extension from the model by E. Nelson \cite{Nelson(1966a),Nelson(2001_book),Nelson(1985_book)}.
Especially, our main purpose by using Nelson's stochastic quantization
is the investigation of {\bf radiation reaction} which is the kicked-back
effect acting on an electron by its radiation \cite{Dirac(1938a)}.
Many works of radiation reaction have been discussed in classical
dynamics from early 1900's, however, the corrections by non-linear
quantum electrodynamics (QED) becomes important in {\bf high-intensity field physics}
produced by the state-of-the-art $O(10\mathrm{PW})$ lasers \cite{Mourou(1997a),ELI-NP,ELI-beams,ELI-ALPS,RA5}.
Comparing the radiation formulas in classical dynamics and non-linear
QED, the factor of $q(\chi)$ can be found \cite{A. Sokolov(1986),Seto(2015a),I. Sokolov(2011a)}:
\[
\frac{dW_{\mathrm{QED}}}{dt}=q(\chi)\times\frac{dW_{\mathrm{classical}}}{dt}
\]
Where, $dW_{\mathrm{classical}}/dt=e^{2}/6\pi\varepsilon_{0}c^{3}\times g_{\alpha\beta}\,dv^{\alpha}/d\tau\,dv^{\beta}/d\tau$
denotes Larmor's formula for the energy loss of radiation in classical
dynamics \cite{Larmor} with respect to the Minkowski metric $g\coloneqq\mathrm{diag}(+1,-1,-1,-1)$
and the 4-velocity $v$. The uniqueness of radiation reaction in high-intensity
field physics is the dependence of the field strength $F_{\mathrm{ex}}$
and $\gamma$ the normalized energy of an electron via this factor
$q(\chi)$, since 
\[
\chi\coloneqq\frac{3}{2}\frac{\hbar}{m_{0}^{2}c^{3}}\sqrt{-g_{\mu\nu}(-eF_{\mathrm{ex}}^{\mu\alpha}v_{\alpha})(-eF_{\mathrm{ex}}^{\nu\beta}v_{\beta})}=O(F_{\mathrm{ex}}\times\gamma)\,.
\]
However as we will discuss the detail in {\bf Volume II} \cite{Vol II},
this is the formula applied only in the case of an electron interacting
with an external laser field of {\bf a plane wave} \cite{Zakowicz(2005),Boca-Florescu(2010)}
in order to the formulation of the non-linear Compton scattering \cite{Brown-Kibble(1964),Nikishov(1964a)),Nikishov(1964b)}.
Thus, it is natural to consider how to generalize it for the case
of focused or superpositioned light. In addition, what is the origin
of this factor $q(\chi)$?

Due to this reason, to identify the origin of $q(\chi)$ under general
external fields is the strong research motivation on the issue of
radiation reaction. For this aspect, we choose Nelson's quantum dynamics
by using Brownian motions in this article series. Nelson's scheme
has a very high-compatibility between equations of motion in classical
and quantum dynamics. Here, let us summarize Nelson's stochastic quantization.

E. Nelson proposed a quantization via a Brownian kinematics. A certain
quanta draws a 3-dimensional Brownian motion as its trajectory in
the non-relativistic regime \cite{Nelson(1966a),Nelson(2001_book),Nelson(1985_book)}.
By employing the kinematics $d\hat{\bm{x}}(t,\omega)=\bm{V}_{\pm}(\hat{\bm{x}}(t,\omega),t)dt+\sqrt{\hbar/2m_{0}}\times d\bm{w}_{\pm}(t,\omega)$
and its Fokker-Planck equations, he succeeded to demonstrate not only
(A) his classical-like dynamics with the sub-equations
\[
m_{0}\left[\partial_{t}\bm{v}(\bm{x},t)+\bm{v}(\bm{x},t)\cdot\nabla\bm{v}(\bm{x},t)-\bm{u}(\bm{x},t)\cdot\nabla\bm{u}(\bm{x},t)-\frac{\hbar}{2m_{0}}\nabla^{2}\bm{u}(\bm{x},t)\right]=-\nabla\mathscr{V}(\bm{x},t)
\]
\[
\bm{v}(\bm{x},t)=\frac{\bm{V}_{+}(\bm{x},t)+\bm{V}_{-}(\bm{x},t)}{2}=\mathrm{Im}\left\{ \frac{\hbar}{m_{0}}\nabla\ln\psi(\bm{x},t)\right\} 
\]
\[
\bm{u}(\bm{x},t)=\frac{\bm{V}_{+}(\bm{x},t)-\bm{V}_{-}(\bm{x},t)}{2}=\mathrm{Re}\left\{ \frac{\hbar}{m_{0}}\nabla\ln\psi(\bm{x},t)\right\} 
\]
is equivalent to the Schr{\"{o}}dinger equation 
\[
i\hbar\partial_{t}\psi(\bm{x},t)=\left[-\frac{\hbar^{2}}{2m_{0}}\nabla^{2}+\mathscr{V}(\bm{x},t)\right]\psi(\bm{x},t)\,,
\]
but also (B) he answered why the square of the wave function should
be regarded as the probability density $\rho(\bm{x},t)=\psi^{*}(\bm{x},t)\psi(\bm{x},t)$
\cite{Nelson(1966a),Nelson(2001_book)}. The biggest advantage of
his method is (C) the ability to draw the ``real'' trajectory of
a quantum particle by a stochastic process. However, the feasibility
of the coupled system between a stochastic particle and fields has
not been established enough. Hence, the realization of the field-generation
mechanism on it is the milestone for the description of radiation
reaction. 

Thus, this first volume dedicates the fundamental construction of
a stochastic kinematics and dynamics of a scalar electron (a Klein-Gordon
particle) including the mechanism of its light emission. Then, {\bf Volume II}
\cite{Vol II} focuses on the formulation of radiation reaction on
its stochastic trajectory.

In {\bf Section \ref{Sect2}} of this paper, we introduce the kinematics
of a scalar electron. At first, we define a stochastic differential
equation $d\hat{x}(\tau,\omega)=\mathcal{V}_{\pm}(\hat{x}(\tau,\omega))d\tau+\lambda\times dW_{\pm}(\tau,\omega)$
as a relativistic kinematics of a scalar electron in the Minkowski
spacetime. Where, $W_{+}(\circ,\bullet)$ and $W_{-}(\circ,\bullet)$
are defined as a special classes of Wiener processes for this relativistic
kinematics. Then, its complex velocity $\mathscr{\mathcal{V}}(\hat{x}(\tau,\omega))\coloneqq(1-i)/2\times\mathcal{V}_{+}(\hat{x}(\tau,\omega))+(1+i)/2\times\mathcal{V}_{-}(\hat{x}(\tau,\omega))$
\cite{Nottale(2011)} is introduced. It can be found that this $\mathscr{\mathcal{V}}$
plays a role of the main cast in the present dynamics of a scalar
electron in {\bf Section \ref{Sect3}}. In order to this kinematics,
its Fokker-Planck equation is derived. In the end of {\bf Section \ref{Sect2}},
we discuss the most delicate problem how to define the proper time
holding the transition between classical physics and quantum physics.

For the calculation of $d\hat{x}(\tau,\omega)=\mathcal{V}_{\pm}(\hat{x}(\tau,\omega))d\tau+\lambda\times dW_{\pm}(\tau,\omega)$,
it requires us to define the evolution of the velocities $\mathcal{V}_{\pm}(\hat{x}(\tau,\omega))$
or $\mathcal{V}(\hat{x}(\tau,\omega))$. Therefore, the dynamics of
a scalar electron interacting with fields is discussed in {\bf Section \ref{Sect3}}.
The mechanism of the field generation from a stochastic particle is
discussed at here, too. Thus, our attention is devoted to the equivalency
between the present model and the well-known system of the Klein-Gordon
equation and the Maxwell equation via an action integral (the functional).
Hereby, the following dynamics of a stochastic scalar electron and
fields is realized.\begin{snugshade}
\[
\begin{gathered}m_{0}\mathfrak{D_{\tau}}\mathcal{V}^{\mu}(\hat{x}(\tau,\omega))=-e\mathcal{\hat{V}}_{\nu}(\hat{x}(\tau,\omega))F^{\mu\nu}(\hat{x}(\tau,\omega))\end{gathered}
\]
\[
\partial_{\mu}\left[F^{\mu\nu}(x)+\delta f^{\mu\nu}\right]=\mu_{0}\times\mathbb{E}\left\llbracket -ec\int_{\mathbb{R}}d\tau\,\mathrm{Re}\left\{ \mathcal{V}^{\nu}(x)\right\} \delta^{4}(x-\hat{x}(\tau,\bullet))\right\rrbracket 
\]
\end{snugshade}

Finally, we summarize this {\bf Volume I} and propose the motivation
to {\bf Volume II} \cite{Vol II} in {\bf Section \ref{Sect4}}.
Where, the transition between quantum and classical dynamics is discussed
by using Ehrenfest's theorem \cite{Ehrenfest} which is main tool
to consider radiation reaction in {\bf Volume II}.

\section{Kinematics of a scalar electron\label{Sect2}}

The first part is a stochastic and relativistic kinematics of a scalar
electron. Let $\mathbb{A}^{4}(\mathbb{V_{\mathrm{M}}^{\mathrm{4}}},g)$
be a 4-dimensional metric affine space with respect to a 4-dimensional
standard vector space $\mathbb{V_{\mathrm{M}}^{\mathrm{4}}}$ and
its $g$ on $\mathbb{V_{\mathrm{M}}^{\mathrm{4}}}$ \cite{Arai(2005)}.
Defining the measure space $(\mathbb{A}^{4}(\mathbb{V_{\mathrm{M}}^{\mathrm{4}}},g),\mathscr{B}(\mathbb{A}^{4}(\mathbb{V_{\mathrm{M}}^{\mathrm{4}}},g)),\mu)$,
we consider this as the Minkowski spacetime. Where, $\mathscr{B}(I)$
denotes the Borel $\sigma$-algebra of a topological space $I$. The
relativistic kinematics of a stochastic scalar electron is characterized
by the complex velocity $\mathcal{V}$, the Fokker-Planck equations
and the proper time $\tau$. 

Let us regard the coordinate mapping $\varphi(x)\coloneqq(x^{0},x^{1},x^{2},x^{3})$
has been already introduced for $\forall x\in\mathbb{V_{\mathrm{M}}^{\mathrm{4}}}$
or $\forall x\in\mathbb{A}^{4}(\mathbb{V_{\mathrm{M}}^{\mathrm{4}}},g)$
with its origin even if we do not declare it explicitly.

\subsection{Stochastic process}

For a certain abstract non-empty set $\varOmega$, consider the probability
space $(\varOmega,\mathcal{F},\mathscr{P})$. Where, $\mathcal{F}$
is a $\sigma$-algebra of $\varOmega$ and $\mathscr{P}$ denotes
the probability measure on the measurable space $(\varOmega,\mathcal{F})$.
Let us define the continuous stochastic process $\hat{x}(\circ,\bullet)\coloneqq\{\hat{x}(\tau,\omega)\in\mathbb{A}^{4}(\mathbb{V_{\mathrm{M}}^{\mathrm{4}}},g)|\tau\in\mathbb{R},\omega\in\varOmega\}$\footnote{Though $\hat{x}(\tau,\omega)$ can be written as $\hat{x}(\tau)$
or $\hat{x}_{\tau}$, however, we dare to write $\omega$ to visualize
the names of each paths explicitly. } as a $\mathscr{B}(\mathbb{R})\times\mathcal{F}/\mathscr{B}(\mathbb{A}^{4}(\mathbb{V_{\mathrm{M}}^{\mathrm{4}}},g))$-measurable
mapping. Where, by considering two measurable spaces $(X,\mathcal{X})$
and $(Y,\mathcal{Y})$, a $\mathcal{X}/\mathcal{Y}$-measurable mapping
$f$ is a mapping $f:X\rightarrow Y$ satisfying $f^{-1}(A)\coloneqq\{x\in X|f(x)\in A\}\subset\mathcal{X}$
for all $A\in\mathcal{Y}$. Our milestone at here is the construction
of the 4-dimensional stochastic differential equation $d\hat{x}^{\mu=0,1,2,3}(\tau,\omega)=\mathcal{V}_{\pm}^{\mu}(\hat{x}(\tau,\omega))d\tau+\lambda\times dW_{\pm}^{\mu}(\tau,\omega)$
as a relativistic kinematics, extended from Nelson's model $d\hat{x}^{i=1,2,3}(t,\omega)=V_{\pm}^{i}(\hat{\bm{x}}(t,\omega),t)dt+\sqrt{\hbar/2m_{0}}\times dw_{\pm}^{i}(t,\omega)$.

\subsubsection{Forward evolution}

Let us start from the definition of the usual 1-dimensional Wiener
process $w(\circ,\bullet)$:
\begin{defn}[Wiener process]
\begin{leftbar}When a 1-dimensional stochastic process $w(\circ,\bullet)\coloneqq\{w(t,\omega)\in\mathbb{R}|t\in[0,\infty),\omega\in\varOmega\}$
satisfies below, it is named a 1-dimensional {\bf Wiener process}
or a 1-dimensional {\bf Brownian motion}. 

(1) $w(0,\omega)=0$ a.s. 

(2) $t\mapsto w(t,\omega)$ is continuous for each $(t,\omega)\in\mathbb{R}\times\varOmega$.

(3) For all times $0=t_{0}<t_{1}<\cdots<t_{t_{n}}$ ($n\in\mathbb{Z}$),
the increments $\{w(t_{i},\omega)-w(t_{i-1},\omega)\}_{i=1}^{n}$
are independent and each of them follows the normal distribution $\{N(0,t_{i}-t_{i-1})\}_{i=1}^{n}$.\end{leftbar}
\end{defn}
For $\left(\mathit{\Omega},\mathcal{F},\mathscr{P}\right)$, let $\mathbb{E}\llbracket\hat{X}(\bullet)\rrbracket$
be the expectation of a random variable $\hat{X}(\bullet)\coloneqq\{\hat{X}(\omega)|\omega\in\varOmega\}$,
i.e., $\mathbb{E}\llbracket\hat{X}(\bullet)\rrbracket\coloneqq\int_{\Omega}d\mathscr{P}(\omega)\,\hat{X}(\omega)$.
The conditional probability of $B$ given $A$ is described by $\mathscr{P}_{A}(B)$.
For $\{A_{n}\}_{n=1}^{\infty}$ of a countable decomposition of $\varOmega$,
its minimum $\sigma$-algebra $\mathcal{\mathscr{C}}=\sigma(\{A_{n}\}_{n=1}^{\infty})$
is introduced. Then, $\mathbb{E}\llbracket\hat{X}(\bullet)|\mathcal{\mathscr{C}}\rrbracket(\omega)\coloneqq\sum_{n=1}^{\infty}\{\int_{\varOmega}d\mathscr{P}_{A_{n}}(\omega')\hat{X}(\omega')\}\mathbf{1}_{A_{n}}(\omega)$
is defined as the conditional expectation of $\hat{X}(\bullet)$ given
$\mathcal{\mathscr{C}}\subset\mathcal{F}$; $\mathbf{1}_{A_{n}}(\omega)$
satisfies $\mathbf{1}_{A_{n}}(\omega\in A_{n})=1$ and $\mathbf{1}_{A_{n}}(\omega\notin A_{n})=0$.

\begin{lem}
\begin{leftbar}Consider $\left(\mathit{\Omega},\mathcal{F},\mathscr{P}\right)$,
a 1-dimensional Wiener process $w(\circ,\bullet)$ satisfies the following
relations for all $t$:
\begin{equation}
\mathbb{E}\llbracket w(t+\delta t,\bullet)-w(t,\bullet)\rrbracket=0
\end{equation}
\begin{equation}
\underset{\delta t\rightarrow0+}{\lim}\mathbb{E}\left\llbracket \frac{[w(t+\delta t,\bullet)-w(t,\bullet)]\times[w(t+\delta t,\bullet)-w(t,\bullet)]}{\delta t}\right\rrbracket =1
\end{equation}
For $i=1,2,\cdots,N$, consider individual 1-dimensional Wiener processes
$w^{i}(\circ,\bullet)$ in each probability spaces $(\varOmega_{i},\mathcal{F}_{i},\mathscr{P}_{i})$.
By $\varOmega\coloneqq\varOmega_{1}\times\cdots\times\varOmega_{N}$,
$\mathcal{F}\coloneqq\mathcal{F}_{1}\otimes\cdots\otimes\mathcal{F}_{N}$
and $\mathscr{P}\coloneqq\mathscr{P}_{1}\times\cdots\times\mathscr{P}_{N}$,
let $\omega\coloneqq(\omega_{1},\cdots,\omega_{N})$ be in $\varOmega$,
an $N$-dimensional Wiener process $w(\circ,\omega)=(w^{1}(\circ,\omega_{1}),\cdots,w^{N}(\circ,\omega_{N}))$
on the probability space $(\varOmega,\mathcal{F},\mathscr{P})$ is
imposed naturally, too. Let us describe it like $w(\circ,\omega)=(w^{1}(\circ,\omega),\cdots w^{N}(\circ,\omega))$
for emphasizing that it is on $(\varOmega,\mathcal{F},\mathscr{P})$.
\begin{equation}
\mathbb{E}\llbracket w^{i}(t+\delta t,\bullet)-w^{i}(t,\bullet)\rrbracket=0
\end{equation}
\begin{equation}
\underset{\delta t\rightarrow0+}{\lim}\mathbb{E}\left\llbracket \frac{[w^{i}(t+\delta t,\bullet)-w^{i}(t,\bullet)]\times[w^{j}(t+\delta t,\bullet)-w^{j}(t,\bullet)]}{\delta t}\right\rrbracket =\delta^{ij}
\end{equation}
Where, $i,j=1,2,\cdots,N$.\end{leftbar}
\end{lem}
 
\begin{defn}[$\{\mathcal{P}_{t}\}$ and $\{\mathcal{F}_{t}\}$]
\begin{leftbar}For $\left(\mathit{\Omega},\mathcal{F},\mathscr{P}\right)$,
$\{\mathcal{P}_{t}\}_{t\in\mathbb{R}}$ is the increasing family of
sub-$\sigma$-algebras of $\mathcal{F}$ such that  $-\infty<s\leq t\Longrightarrow\mathcal{P}_{s}\subset\mathcal{P}_{t}$.
And $\{\mathcal{F}_{t}\}_{t\in\mathbb{R}}$ is the decreasing family
 fulfilling $t\leq s<\infty\Longrightarrow\mathcal{F}_{t}\supset\mathcal{F}_{s}$.
\end{leftbar}
\end{defn}
For $(\varOmega,\mathcal{F},\mathscr{P})$, consider a family of $\mathscr{B}([0,t])\times\mathcal{F}/\mathscr{B}(X)$-measurable
mappings ($X$ is a $N$-dimensional topological space)
\[
\mathbb{L}_{T}^{p}(X)\coloneqq\left\{ f(\circ,\bullet)\left|\int_{T\subset\mathbb{R}}|\varphi^{i}\circ f(t,\omega)|^{p}dt<\infty\,\mathrm{a.s.},\,i=1,2,\cdots,N\right.\right\} 
\]
where, $\{\varphi^{i}\}_{i=1}^{N}$ is the coordinate mapping $\varphi^{i}:X\rightarrow\mathbb{R}$.
Then, its ''adapted'' class is,
\[
\mathcal{L}_{\mathrm{loc}}^{p}(\{\mathcal{P}_{t}\};X)\coloneqq\left\{ \left.f(\circ,\bullet)\in\mathbb{L}_{[0,t]}^{p}(X)\right|f(\circ,\bullet)\mathrm{\,is\,}\{\mathcal{P}_{t}\}\mathchar`-\mathrm{adapted}\right\} \,.
\]
Where, $\{\mathcal{P}_{t}\}$-adapted means a stochastic process which
is $\mathcal{P}_{t}/\mathscr{B}(X)$-measurable for all $t$. \begin{leftbar}
\begin{defn}[$\{\mathcal{P}_{t}\}$-Wiener process]
When an $N$-dimensional Wiener process \textbf{$w(\circ,\omega)$}
satisfies the following, it is called a $\{\mathcal{P}_{t}\}$-Wiener
process.

(1) $w(\circ,\bullet)$ is $\{\mathcal{P}_{t}\}$-adapted.

(2) $w(t,\omega)-w(s,\omega)$ and $\mathcal{P}_{s}$ are independent
for $0\leq\forall s\le t$.
\end{defn}
\end{leftbar}\begin{leftbar}
\begin{defn}[It\^{o} integral]
The integral,
\[
\int_{0}^{t}f(t',\omega)dw^{i}(t',\omega),\qquad i=1,2,\cdots,N
\]
is defined as an It\^{o} integral of a function $f\in\mathcal{L}_{\mathrm{loc}}^{2}(\{\mathcal{P}_{t}\};\mathbb{R})$
for a $\{\mathcal{P}_{t}\}$-Wiener process $w(\circ,\omega)$ \cite{Ito(1944)}.
\end{defn}
\end{leftbar}Let us use $w_{+}(\circ,\omega)$ as a $\{\mathcal{P}_{t}\}$-Wiener
process from here.
\begin{lem}
\begin{leftbar}For  $f(\circ,\bullet),g(\circ,\bullet)\in\mathcal{L}_{\mathrm{loc}}^{2}(\{\mathcal{P}_{t};\mathbb{R}\})$,
the following formula is imposed \cite{Ito(1944)}: 
\begin{equation}
\mathbb{E}\left\llbracket \int_{0}^{\tau}f(t',\bullet)dw_{+}^{i}(t',\bullet)\cdot\int_{0}^{\tau}g(t'',\bullet)dw_{+}^{j}(t'',\bullet)\right\rrbracket =\delta^{ij}\times\mathbb{E}\left\llbracket \int_{0}^{\tau}f(t',\bullet)\cdot g(t',\bullet)dt'\right\rrbracket 
\end{equation}
\end{leftbar}\end{lem}
\begin{thm}[It\^{o} formula]
\begin{leftbar}For $a^{i}(\circ,\bullet)\in\mathcal{L}_{\mathrm{loc}}^{1}(\{\mathcal{P}_{t}\};\mathbb{R})$
and $b_{\alpha}^{i}(\circ,\bullet)\in\mathcal{L}_{\mathrm{loc}}^{2}(\{\mathcal{P}_{t}\};\mathbb{R})$
$(1\leqq i\leqq N,\,1\leqq\alpha\leqq d)$, consider an $N$-dimensional
 process $\hat{X}(\circ,\omega)$ given by
\begin{equation}
\hat{X}^{i}(t,\omega)=\hat{X}^{i}(0,\omega)+\int_{0}^{t}a^{i}(t',\omega)dt'+\sum_{\alpha=1}^{d}\int_{0}^{t}b_{\alpha}^{i}(t',\omega)dw_{+}^{\alpha}(t',\omega)\label{eq: Ito-process +}
\end{equation}
with respect to an initial condition $\hat{X}^{i}(0,\omega)$. Then
for a function $f\in C^{2}(\mathbb{R}^{N})$, 
\begin{eqnarray}
f(\hat{X}(t,\omega)) & = & f(\hat{X}(0,\omega))+\sum_{i=1}^{N}\int_{0}^{t}\frac{\partial f}{\partial x^{i}}(\hat{X}(t',\omega))\left[a^{i}(t',\omega)dt'+\sum_{\alpha=1}^{d}b_{\alpha}^{i}(t',\omega)dw_{+}^{\alpha}(t',\omega)\right]\nonumber \\
 &  & +\frac{1}{2}\sum_{i,j=1}^{N}\sum_{\alpha=1}^{d}\int_{0}^{t}\frac{\partial^{2}f}{\partial x^{i}\partial x^{j}}(\hat{X}(t',\omega))b_{\alpha}^{i}(t',\omega)b_{\alpha}^{j}(t',\omega)dt'\label{eq: Ito-formula-1}
\end{eqnarray}
is almost surely satisfied. By introducing $\int_{0}^{t}d_{+}\hat{X}(t',\omega)\coloneqq\hat{X}(t,\omega)-\hat{X}(0,\omega)$
and $\int_{0}^{t}d_{+}f(\hat{X}(t',\omega))\coloneqq f(\hat{X}(t,\omega))-f(\hat{X}(0,\omega))$,
(\ref{eq: Ito-formula-1}) is written like
\begin{eqnarray}
d_{+}f(\hat{X}(t,\omega)) & = & \sum_{i=1}^{N}\frac{\partial f}{\partial x^{i}}(\hat{X}(t,\omega))d_{+}\hat{X}^{i}(t,\omega)+\frac{1}{2}\sum_{i,j=1}^{N}\frac{\partial^{2}f}{\partial x^{i}\partial x^{j}}(\hat{X}(t,\omega))d_{+}\hat{X}^{i}(t,\omega)d_{+}\hat{X}^{j}(t,\omega)\label{eq: Ito-formula-2}
\end{eqnarray}
\end{leftbar}
\end{thm}
Though a $\{\mathcal{P}_{t}\}$-Wiener process is defined on $[0,\infty)$,
however, it can be expand on to $\mathbb{R}$ naturally. For example,
let us define a new stochastic process like $w'_{+}(t-T,\omega)\coloneqq w_{+}(t,\omega)$
with respect to $T>0$. This new $w'_{+}(\circ,\omega)$ is a $\{\mathcal{P}_{t}\}$-Wiener
process such that $w'_{+}(-T,\omega)=0$ a.s. By defining another
process $w''_{+}(t-T,\omega)\coloneqq w'_{+}(t-T,\omega)-w'_{+}(0,\omega)$,
this $w''_{+}(\circ,\omega)$ is a $\{\mathcal{P}_{t}\}$-Wiener process
on $[-T,\infty)$, satisfying $w''_{+}(0,\omega)=0$ a.s. By repeating
this procedure, the domain of a $\{\mathcal{P}_{t}\}$-Wiener process
$w_{+}(\circ,\omega)$ can be expanded on to $\mathbb{R}$ with the
condition $w_{+}(0,\omega)=0$ a.s. Therefore, the differential form
of (\ref{eq: Ito-formula-2}) is imposed for all $t\in\mathbb{R}$.
From here, we select the domain of stochastic processes on $\mathbb{R}$
for the later discussion. In order to this reason, $\mathcal{L}_{\mathrm{loc}}^{p}(\{\mathcal{P}_{t}\};X)$
is replaced by the new definition
\[
\mathcal{L}_{\mathrm{loc}}^{p}(\{\mathcal{P}_{t}\};X)\coloneqq\left\{ \left.f(\circ,\bullet)\in\mathbb{L}_{(-\infty,t]}^{p}(X)\right|f(\circ,\bullet)\mathrm{\,is\,}\{\mathcal{P}_{t}\}\mathchar`-\mathrm{adapted}\right\} \,.
\]
If a stochastic process $f(\circ,\bullet)$ is $\mathscr{B}((-\infty,t])\times\mathcal{P}_{t}/\mathscr{B}(X)$-measurable
for all $t\in\mathbb{R}$, let us name $f(\circ,\bullet)$ $\{\mathcal{P}_{t}\}$-progressively
measurable or $\{\mathcal{P}_{t}\}$-progressive. Hence, $\hat{X}(\circ,\bullet)$
of (\ref{eq: Ito-process +}) is $\{\mathcal{P}_{t}\}$-progressive.

Consider a simple 1-dimentional $\{\mathcal{P}_{t}\}$-progressive
$\hat{X}(\circ,\bullet)$ such that $d_{+}\hat{X}(t,\omega)=a_{+}(\hat{X}(t,\omega))dt+\theta_{+}\times dw_{+}(t,\omega)$
($\theta_{+}>0$) and $d\mathbb{E}\llbracket f(\hat{X}(t,\bullet))\rrbracket/dt$
for a function $f\in C^{2}(\mathbb{R})$. Since $\mathbb{E}\llbracket f(\hat{X}(t,\bullet))\rrbracket\coloneqq\int_{\mathbb{R}}dx\,f(x)p(x,t)$
for the probability density $p(x,t)$, thus, $d\mathbb{E}\llbracket f(\hat{X}(t,\bullet))\rrbracket/dt=\int_{\mathbb{R}}dx\,f(x)\partial_{t}p(x,t)$
and also
\begin{eqnarray}
\frac{d\mathbb{E}\llbracket f(\hat{X}(t,\bullet))\rrbracket}{dt} & = & \lim_{\varDelta t\rightarrow0+}\frac{\mathbb{E}\llbracket f(\hat{X}(t+\varDelta t,\bullet))\rrbracket-\mathbb{E}\llbracket f(\hat{X}(t,\bullet))\rrbracket}{\varDelta t}\nonumber \\
 & = & \lim_{\varDelta t\rightarrow0+}\mathbb{E}\left\llbracket \mathbb{E}\left\llbracket \left.\frac{1}{\varDelta t}\begin{gathered}\int_{t}^{t+\varDelta t}d_{+}f(\hat{X}(t,\bullet))\end{gathered}
\right|\mathcal{P}_{t}\right\rrbracket \right\rrbracket \,,
\end{eqnarray}
where, $\mathbb{E}\llbracket f(\hat{X}(\tau,\bullet))\rrbracket=\mathbb{E}\llbracket\mathbb{E}\llbracket f(\hat{X}(\tau,\bullet))|\mathcal{P}_{t}\rrbracket\rrbracket$
is employed. Then, the Fokker-Planck equation is derived:
\begin{equation}
\partial_{t}p(x,t)+\partial_{x}[a_{+}(x)p(x,t)]=\frac{\theta_{+}^{2}}{2}\partial_{x}^{2}p(x,t)
\end{equation}

\subsubsection{Backward evolution}

Let us consider a type of a backward evolution. The point at here
is to consider a decreasing family of $\{\mathcal{F}_{t}\}$. 
\begin{lem}
\begin{leftbar}An $N$-dimensional $\{\mathcal{F}_{t}\}$-Wiener
process $w_{-}(\circ,\bullet)$ fulfills the following relations ($i,j=1,2,\cdots,N$):
\begin{equation}
\mathbb{E}\llbracket w_{-}^{i}(t,\bullet)-w_{-}^{i}(t-\delta t,\bullet)\rrbracket=0
\end{equation}
\begin{equation}
\underset{\delta t\rightarrow0+}{\lim}\mathbb{E}\left\llbracket \frac{[w_{-}^{i}(t,\bullet)-w_{-}^{i}(t-\delta t,\bullet)]\times[w_{-}^{j}(t,\bullet)-w_{-}^{j}(t-\delta t,\bullet)]}{\delta t}\right\rrbracket =\delta^{ij}
\end{equation}
\end{leftbar}
\end{lem}
Then, the following function family is defined:
\[
\mathcal{L}_{\mathrm{loc}}^{p}(\{\mathcal{F}_{t}\};X)\coloneqq\left\{ \left.f(\circ,\bullet)\in\mathbb{L}_{[t,\infty)}^{p}(X)\right|f(\circ,\bullet)\mathrm{\,is\,}\{\mathcal{F}_{t}\}\mathchar`-\mathrm{adapted}\right\} 
\]
If a stochastic process $f(\circ,\bullet)$ is $\mathscr{B}([t,\infty))\times\mathcal{F}_{t}/\mathscr{B}(X)$-measurable
for all $t\in\mathbb{R}$, let us call it $\{\mathcal{F}_{t}\}$-progressively
measurable or $\{\mathcal{F}_{t}\}$-progressive.
\begin{thm}
\begin{leftbar}For $t<\forall T<\infty$, $\{A^{i}(\circ,\bullet)\}\in\mathcal{L}_{\mathrm{loc}}^{1}(\{\mathcal{F}_{t}\};\mathbb{R})$,
$\{B_{\alpha}^{i}(\circ,\bullet)\}\in\mathcal{L}_{\mathrm{loc}}^{2}(\{\mathcal{F}_{t}\};\mathbb{R})$
and a $d$-dimensional $\{\mathcal{F}_{t}\}$-Wiener process $\{w_{-}^{\alpha}(\circ,\omega)\}$
$(1\leqq i\leqq N,\,1\leqq\alpha\leqq d)$, consider an $N$-dimensional
 process ($\{\mathcal{F}_{t}\}$-progressive) $\hat{X}(\circ,\omega)$
given by
\begin{equation}
\hat{X}^{i}(t,\omega)=\hat{X}^{i}(T,\omega)-\int_{t}^{T}A^{i}(t',\omega)dt'-\sum_{\alpha=1}^{d}\int_{t}^{T}B_{\alpha}^{i}(t',\omega)dw_{-}^{\alpha}(t',\omega)
\end{equation}
with respect to a terminal condition $\hat{X}(T,\omega)$. By introducing
$\int_{t}^{T}d_{-}\hat{X}(t',\omega)\coloneqq\hat{X}(T,\omega)-\hat{X}(t,\omega)$,
i.e.,
\begin{equation}
d_{-}\hat{X}^{i}(t,\omega)=A^{i}(t,\omega)dt+\sum_{\alpha=1}^{d}B_{\alpha}^{i}(t,\omega)dw_{-}^{\alpha}(t,\omega)
\end{equation}
and $\int_{t}^{T}d_{-}f(\hat{X}(t',\omega))\coloneqq f(\hat{X}(T,\omega))-f(\hat{X}(t,\omega))$
for a function $f\in C^{2}(\mathbb{R}^{N})$, 
\begin{equation}
d_{-}f(\hat{X}(t,\omega))=\sum_{i=1}^{N}\frac{\partial f}{\partial x^{i}}(\hat{X}(t,\omega))d_{-}\hat{X}^{i}(t,\omega)-\frac{1}{2}\sum_{i,j=1}^{N}\frac{\partial^{2}f}{\partial x^{i}\partial x^{j}}(\hat{X}(t,\omega))d_{-}\hat{X}^{i}(t,\omega)d_{-}\hat{X}^{j}(t,\omega)\,\,\mathrm{a.s.}\label{eq: Ito-formula-3}
\end{equation}
\end{leftbar}
\end{thm}
Consider the derivation of (\ref{eq: Ito-formula-3}) by using (\ref{eq: Ito-formula-2}).
The decreasing family $\{\mathcal{F}_{t}\}_{t\in\mathbb{R}}$ relates
to an increasing family $\{\mathcal{P}_{t'}\}_{t'\in\mathbb{R}}$
by the replacement of its subscripts: By defining a monotonically
decreasing function $f:\mathbb{R}\rightarrow\mathbb{R}$, then, $\{\mathcal{F}_{f(t)}\}_{t\in\mathbb{R}}$
becomes an increasing family for $t$ since $t\le s\Rightarrow\mathcal{F}_{f(t)}\subset\mathcal{F}_{f(s)}$.
Thus, there is a correspondence between each of $t$ and $t'$. $\{\mathcal{P}_{t'}\}_{t'\in\mathbb{R}}\equiv\{\mathcal{F}_{f(t)}\}_{t\in\mathbb{R}}$.
For an $\{\mathcal{F}_{t}\}$-adapted $\hat{X}(\circ,\bullet)$, there
is a $\{\mathcal{P}_{t'}\}$-adapted $\hat{X}'(\circ,\bullet)$ satisfying
$\hat{X}'(t',\omega)=\hat{X}(t,\omega)$ at a fixed $t$. In order
to the construction of $\{\mathcal{P}_{t'}\}$ and $\hat{X}'(\circ,\bullet)$,
$\hat{X}'(t'+T,\omega)=\hat{X}(t-T,\omega)$ for all $T\in\mathbb{R}$
via the relation $f(t-T)\coloneqq t'+T$. Since $\int_{t'-T}^{t'+T}d_{+}\hat{X}'(\mathring{t},\omega)=\hat{X}'(t'+T,\omega)-\hat{X}'(t'-T,\omega)=\hat{X}(t-T,\omega)-\hat{X}(t+T,\omega)=-\int_{t-T}^{t+T}d_{-}\hat{X}(\mathring{t},\omega)$
for $T>0$, therefore, $d_{+}\hat{X}'(t',\omega)=-d_{-}\hat{X}(t,\omega)$
is derived. And $d_{+}f(\hat{X}'(t',\omega))=-d_{-}f(\hat{X}(t,\omega))$
is also satisfied for a function $f\in C^{2}(\mathbb{R}^{N})$ by
the same way. Let us employ those relations on (\ref{eq: Ito-formula-2})
namely, 
\begin{eqnarray}
d_{+}f(\hat{X}'(t',\omega)) & = & \sum_{i=1}^{N}\frac{\partial f}{\partial x^{i}}(\hat{X}'(t',\omega))d_{+}\hat{X}'^{i}(t',\omega)\nonumber \\
 &  & +\frac{1}{2}\sum_{i,j=1}^{N}\frac{\partial^{2}f}{\partial x^{i}\partial x^{j}}(\hat{X}'(t',\omega))d_{+}\hat{X'}^{i}(t',\omega)d_{+}\hat{X}'^{j}(t',\omega)\,\,\mathrm{a.s.},
\end{eqnarray}
the It\^{o} formula of (\ref{eq: Ito-formula-3}) is imposed by switching
from the $\{\mathcal{P}_{t'}\}$-adapted $\hat{X}'(\circ,\bullet)$
to the $\{\mathcal{F}_{t}\}$-adapted $\hat{X}(\circ,\bullet)$. This
discussion can be applied for each $t\in\mathbb{R}$.

In the case of a 1-dimensional $\{\mathcal{F}_{t}\}$-progressive
$\hat{Y}(\circ,\bullet)$, i.e., $d_{-}\hat{Y}(t,\omega)=a_{-}(\hat{Y}(t,\omega))dt+\theta_{-}\times dw_{-}(t,\omega)$
($\theta_{-}>0$), since 
\begin{eqnarray}
\frac{d\mathbb{E}\llbracket f(\hat{Y}(t,\bullet))\rrbracket}{dt} & = & \lim_{\varDelta t\rightarrow0+}\frac{\mathbb{E}\llbracket f(\hat{Y}(t,\bullet))\rrbracket-\mathbb{E}\llbracket f(\hat{Y}(t-\varDelta t,\bullet))\rrbracket}{\varDelta t}\nonumber \\
 & = & \lim_{\varDelta t\rightarrow0+}\mathbb{E}\left\llbracket \mathbb{E}\left\llbracket \left.\frac{1}{\varDelta t}\begin{gathered}\int_{t-\varDelta t}^{t}d_{-}f(\hat{X}(t',\bullet))\end{gathered}
\right|\mathcal{F}_{t}\right\rrbracket \right\rrbracket \,,
\end{eqnarray}
the following Fokker-Planck equation is fulfilled:
\begin{equation}
\partial_{t}p(y,t)+\partial_{y}[a_{-}(y)p(y,t)]=-\frac{\theta_{-}^{2}}{2}\partial_{y}^{2}p(y,t)
\end{equation}

\subsubsection{Forward-backward composition}

Let us rewrite 1-dimensional stochastic processes of $\{\mathcal{P}_{t}\}$-adapted
$\hat{X}(\circ,\bullet)$ and $\{\mathcal{F}_{t}\}$-adapted $\hat{Y}(\circ,\bullet)$
on $T\subset\mathbb{R}$:
\begin{equation}
\begin{cases}
\begin{gathered}\int_{t\in T}d_{+}\hat{X}(t,\omega)=\int_{t\in T}a_{+}(\hat{X}(t,\omega))dt+\theta\int_{t\in T}dw_{+}(t,\omega)\end{gathered}
, & \{\mathcal{P}_{t}\}\mathchar`-\mathrm{adapted}\\
\begin{gathered}\int_{t\in T}d_{-}\hat{Y}(t,\omega')=\int_{t\in T}a_{-}(\hat{Y}(t,\omega'))dt+\theta\int_{t\in T}dw_{-}(t,\omega')\end{gathered}
, & \{\mathcal{F}_{t}\}\mathchar`-\mathrm{adapted}
\end{cases}\label{eq: set of SDE eqs 1}
\end{equation}
Where, the following It\^{o} rule is satisfied; $dt\cdot dt=0$,
$dt\cdot dw_{\pm}(t,\omega)=0$, $dw_{\pm}(t,\omega)\cdot dw_{\pm}(t,\omega)=dt$
and $dw_{\pm}(t,\omega)\cdot dw_{\mp}(t,\omega)=0$. By describing
(\ref{eq: set of SDE eqs 1}) like 
\begin{equation}
\left[\begin{array}{c}
\hat{X}(b,\omega)-\hat{X}(a,\omega)\\
\hat{Y}(b,\omega)-\hat{Y}(a,\omega)
\end{array}\right]=\int_{a}^{b}\left[\begin{array}{c}
a_{+}(\hat{X}(t,\omega),\hat{Y}(t,\omega))\\
a_{-}(\hat{X}(t,\omega),\hat{Y}(t,\omega))
\end{array}\right]dt+\theta\times\int_{a}^{b}\left[\begin{array}{c}
dw_{+}(t,\omega)\\
dw_{-}(t,\omega)
\end{array}\right]\label{eq: set of SDE eqs 2}
\end{equation}
for $a_{+},a_{-}:\mathbb{R}^{2}\rightarrow\mathbb{R}$, it can be
regarded as the equation of a 2-component vector $\hat{\gamma}(t,\omega)\coloneqq(\hat{X}(t,\omega),\hat{Y}(t,\omega))$.
Though, $\{\mathcal{P}_{t}\}$ and $\{\mathcal{F}_{t}\}$ are the
different types of the families of sub-$\sigma$-algebras, however,
the index of $t\in\mathbb{R}$ is common. Therefore, the mathematical
set of $(\mathcal{P}_{t},\mathcal{F}_{t})$ can be defined for each
$t$. Let us regard this as a new sub-$\sigma$-algebra for a 2-dimensional
stochastic process denoted by $\mathcal{M}_{t\in\mathbb{R}}\coloneqq\mathcal{P}_{t}\otimes\mathcal{F}_{t}\in\mathcal{F}$
and consider its family $\{\mathcal{M}_{t}\}_{t\in\mathbb{R}}$. 
\begin{defn}[Forward-backward composition]
\begin{leftbar} Let $\{\mathcal{M}_{t}\}_{t\in\mathbb{R}}$ be a
family of a sub-$\sigma$-algebras for a 2-dimensional stochastic
process $\hat{\gamma}(t,\bullet)\coloneqq(\hat{X}(t,\bullet),\hat{Y}(t,\bullet))$
on $\left(\mathit{\Omega},\mathcal{F},\mathscr{P}\right)$. Then,
$\hat{\gamma}(t,\bullet)$ is called $\mathcal{M}_{t}/\mathscr{B}(\mathbb{R}^{2})$-measurable
or $\mathcal{P}_{t}\otimes\mathcal{F}_{t}/\mathscr{B}(\mathbb{R}^{2})$-measurable,
when $\hat{X}(t,\bullet)$ is $\mathcal{P}_{t}/\mathscr{B}(\mathbb{R})$-measurable
and $\hat{Y}(t,\bullet)$ is $\mathcal{F}_{t}/\mathscr{B}(\mathbb{R})$-measurable.
If $\{\hat{\gamma}(t,\bullet)\}_{t\in\mathbb{R}}$ is $\mathcal{M}_{t}/\mathscr{B}(\mathbb{R}^{2})$-measurable
for all $t$, $\hat{\gamma}(\circ,\bullet)$ is called $\{\mathcal{M}_{t}\}$-adapted
or $\{\mathcal{P}_{t}\otimes\mathcal{F}_{t}\}$-adapted. On the contrary
when a 2-dimensional $\{\mathcal{M}_{t}\}$-adapted process is given,
it can be decomposed into 1-dimensional $\{\mathcal{P}_{t}\}$ and
$\{\mathcal{F}_{t}\}$-adapted processes. \end{leftbar}
\end{defn}
For a $\mathscr{B}(\mathbb{R}^{2})/\mathscr{B}(X)$-measurable function
$f:\mathbb{R}^{2}\rightarrow X$ for a certain topological space $X$,
$f(\hat{\gamma}(t,\bullet))$ is $\mathcal{M}_{t}/\mathscr{B}(X)$-measurable.
Hence, $a(\hat{\gamma}(t,\bullet))\coloneqq(a_{+}(\hat{\gamma}(t,\bullet),a_{-}(\hat{\gamma}(t,\bullet))$
is $\mathcal{M}_{t}/\mathscr{B}(\mathbb{R}^{2})$-measurable. Then,
consider an It\^{o} formula of a function $f\in C^{2}(\mathbb{R}^{2})$
on $\hat{\gamma}(\circ,\omega)$. 
\begin{eqnarray}
f(\hat{\gamma}(b,\omega))-f(\hat{\gamma}(a,\omega)) & = & \int_{a}^{b}\left[\frac{\partial f}{\partial x}(\hat{\gamma}(t,\omega))d_{+}\hat{X}(\hat{\gamma}(t,\omega))+\frac{\partial f}{\partial y}(\hat{\gamma}(t,\omega))d_{-}\hat{Y}(\hat{\gamma}(t,\omega))\right]\nonumber \\
 &  & +\frac{\theta^{2}}{2}\int_{a}^{b}\left[\frac{\partial^{2}f}{\partial x\partial x}(\hat{\gamma}(t,\omega))-\frac{\partial^{2}f}{\partial y\partial y}(\hat{\gamma}(t,\omega))\right]dt\,\,\mathrm{a.s.}
\end{eqnarray}

As the next step, consider a $\{\mathcal{F}_{t}\otimes\mathcal{P}_{t}\otimes\mathcal{P}_{t}\otimes\mathcal{P}_{t}\}$-adapted
$\hat{x}(\circ,\bullet)$, namely,
\begin{equation}
d_{[i]}\hat{x}^{i=0,1,2,3}(t,\omega)\coloneqq a^{i}(\hat{x}(t,\omega))dt+\theta\times dw^{i}(t,\omega)\,,\label{eq: mod-Wiener}
\end{equation}
Where, $\hat{x}^{0}(\circ,\bullet)$ is $\{\mathcal{F}_{t}\}$-adapted
and $\{\hat{x}^{i}(\circ,\bullet)\}_{i=1,2,3}$ are $\{\mathcal{P}_{t}\}$-adapted.
Hence, $\{w^{i}(\circ,\bullet)\}_{i=0,1,2,3}$ is the set of a 1-dimensional
$\{\mathcal{F}_{t}\}$-Wiener process $w^{0}(\circ,\bullet)$ and
a 3-dimensional $\{\mathcal{P}_{t}\}$-Wiener process $\{w^{i}(\circ,\bullet)\}_{i=1,2,3}$
(see its characteristics in {\bf Definition \ref{(-g)-Wiener}} and
{\bf Theorem \ref{(-g)-Wiener-Ito}}). Then, the It\^{o} formula
of a function $f\in C^{2}(\mathbb{R}^{2})$ with respect to $\hat{x}(\circ,\omega)$
is,
\begin{eqnarray}
f(\hat{x}(t,\omega)) & = & f(\hat{x}(t_{0},\omega))+\sum_{i=0,1,2,3}\int_{t_{0}}^{t}\partial_{i}f(\hat{x}(t',\omega))\cdot d_{[i]}\hat{x}^{i}(t',\omega)\nonumber \\
 &  & -\frac{\theta^{2}}{2}\int_{t_{0}}^{t}\left[\partial_{0}^{2}-\partial_{1}^{2}-\partial_{2}^{2}-\partial_{3}^{2}\right]f(\hat{x}(t',\omega))dt'\,\,\mathrm{a.s.}
\end{eqnarray}

\begin{defn}[$(-g)$-Wiener process]
\begin{leftbar}\label{(-g)-Wiener}When $a=0$ and $\theta=1$ on
(\ref{eq: mod-Wiener}), namely, the solution of 
\begin{equation}
d_{[i]}\hat{x}^{i=0,1,2,3}(t,\omega)=dw^{i}(t,\omega)\,,
\end{equation}
is called a $\{\mathcal{F}_{t}\otimes\mathcal{P}_{t}\otimes\mathcal{P}_{t}\otimes\mathcal{P}_{t}\}$-$(-g)$-Wiener
process denoted by $w_{(-g)}(\circ,\bullet)$ such that $w_{(-g)}^{0}(\circ,\bullet)$
is a 1-dimensional $\{\mathcal{F}_{t}\}$-Wiener process and $\{w_{(-g)}^{i}(\circ,\bullet)\}_{i=1,2,3}$
is a 3-dimensional $\{\mathcal{P}_{t}\}$-Wiener process, satisfying
the following for $i,j=0,1,2,3$:
\begin{equation}
\begin{array}{cc}
\mathbb{E}\left\llbracket \int_{t}^{t+\delta t}dw_{(-g)}^{i}(t',\bullet)\right\rrbracket =0\,, & \mathbb{E}\left\llbracket \int_{t}^{t+\delta t}dw_{(-g)}^{i}(t',\bullet)\times\int_{t}^{t+\delta t}dw_{(-g)}^{j}(t'',\bullet)\right\rrbracket =\delta^{ij}\times\delta t\end{array}
\end{equation}
\end{leftbar}
\end{defn}
The name of ''$(-g)$'' derives from the following characteristics:
\begin{thm}[It\^{o} formula and Fokker-Planck equation on $w_{(-g)}(\circ,\bullet)$]
\begin{leftbar}\label{(-g)-Wiener-Ito} For $g=\mathrm{diag}(+1,-1,-1,-1)$,
$w_{(-g)}(\circ,\bullet)$ induces the It\^{o} formula of a function
$f\in C^{2}(\mathbb{R}^{4})$
\begin{equation}
df(w_{(-g)}(t,\omega))=\sum_{i=0}^{3}\frac{\partial}{\partial x^{i}}f(w_{(-g)}(t,\omega))\cdot dw_{(-g)}^{i}(t,\omega)+\frac{\lambda^{2}}{2}\sum_{i,j=0}^{3}(-g^{ij})\frac{\partial^{2}}{\partial x^{i}\partial x^{j}}f(w_{(-g)}(\tau,\omega))\cdot d\tau\,\,\mathrm{a.s.}
\end{equation}
and the following Fokker-Planck equation for its probability density
$p\in C^{2,1}(\mathbb{R}^{4}\times\mathbb{R})$
\begin{equation}
\partial_{t}p(x,t)=\frac{1}{2}\sum_{i,j=0}^{3}(-g^{ij})\frac{\partial^{2}}{\partial x^{i}\partial x^{j}}p(x,t)\,.
\end{equation}
\end{leftbar}
\end{thm}
The following is the basic idea to describe a relativistic kinematics.\begin{leftbar}
\begin{defn}[$\{\mathscr{P}_{\tau}\}$ and $\{\mathscr{F}_{\tau}\}$]
\label {P_and_F}For $\left(\mathit{\Omega},\mathcal{F},\mathscr{P}\right)$,
consider $\{\mathscr{P}_{\tau}\}_{\tau\in\mathbb{R}}$ a family of
composited sub-$\sigma$-algebras $\mathscr{P}_{\tau\in\mathbb{R}}\coloneqq\mathcal{F}_{t}\otimes\mathcal{P}_{t}\otimes\mathcal{P}_{t}\otimes\mathcal{P}_{t}\in\mathcal{F}$
for a 4-dimensional stochastic process. By introducing a monotonically
decreasing function $f:\mathbb{R}\rightarrow\mathbb{R}$, $\{\mathcal{F}_{f(\tau)}\otimes\mathcal{P}_{f(\tau)}\otimes\mathcal{P}_{f(\tau)}\otimes\mathcal{P}_{f(\tau)}\}_{\tau\in\mathbb{R}}$
imposes $\{\mathscr{F}_{\tau}\}_{\tau\in\mathbb{R}}$, a family of
$\mathscr{F}_{\tau\in\mathbb{R}}\coloneqq\mathcal{P}_{\tau}\otimes\mathcal{F}_{\tau}\otimes\mathcal{F}_{\tau}\otimes\mathcal{F}_{\tau}\in\mathcal{F}$.
\end{defn}
\end{leftbar}\begin{leftbar}
\begin{defn}[$W_{+}(\circ,\bullet)$ and $W_{-}(\circ,\bullet)$]
Let $W_{+}(\circ,\bullet)$ be a $\{\mathscr{P}_{\tau}\}$-$(-g)$-Wiener
process, namely, $w_{(-g)}(\circ,\bullet)$. Then, an $\{\mathscr{F}_{\tau}\}$-$(-g)$
Wiener process $W_{-}(\circ,\bullet)$ is imposed by {\bf Definition \ref {P_and_F}}.
Its satisfies the following relations for $\mu,\nu=0,1,2,3$ and a
function $f\in C^{2}(\mathbb{R}^{4})$:
\begin{equation}
\begin{array}{cc}
\mathbb{E}\left\llbracket \int_{\tau}^{\tau+\delta\tau}dW_{\pm}^{\mu}(\tau',\bullet)\right\rrbracket =0\,, & \mathbb{E}\left\llbracket \int_{\tau}^{\tau+\delta\tau}dW_{\pm}^{\mu}(\tau',\bullet)\times\int_{\tau}^{\tau+\delta\tau}dW_{\pm}^{\mu}(\tau'',\bullet)\right\rrbracket =\delta^{\mu\nu}\times\delta\tau\,,\end{array}
\end{equation}
\begin{equation}
f(W_{\pm}(\tau_{b},\omega))-f(W_{\pm}(\tau_{a},\omega))=\int_{\tau_{a}}^{\tau_{b}}\partial_{\mu}f(W_{\pm}(\tau,\omega))dW_{\pm}^{\mu}(\tau,\omega)\mp\frac{\lambda^{2}}{2}\int_{\tau_{a}}^{\tau_{b}}\partial_{\mu}\partial^{\mu}f(W_{\pm}(\tau,\omega))d\tau\,\,\mathrm{a.s.}
\end{equation}

\end{defn}
\end{leftbar}

\subsection{Dual-progressively measurable process $\hat{x}(\circ,\bullet)$}

Let $\hat{x}(\circ,\bullet)$ be a stochastic process as a relativistic
kinematics extended from 3-dimensional Nelson's (R0), (R1), (R2),
(R3), (S1), (S2) and (S2)-processes to 4-dimensional one.  For our
convenience, the coordinate mappings $\{\varphi_{E}^{A}\}$ is introduced
again, such that the index $A$ becomes $A=\mu$ if $E=\mathbb{V}_{\mathrm{M}}^{\mathrm{4}}$
or $E=\mathbb{A}^{4}(\mathbb{V_{\mathrm{M}}^{\mathrm{4}}},g)$ with
its origin and $A=\mu\nu$ when $E=\mathbb{V}_{\mathrm{M}}^{\mathrm{4}}\otimes\mathbb{V}_{\mathrm{M}}^{\mathrm{4}}$
($\mu,\nu=0,1,2,3$):
\begin{defn}[(R0)-process]
\begin{leftbar}For $\left(\mathit{\Omega},\mathcal{F},\mathscr{P}\right)$,
a $\mathscr{B}(\mathbb{R})\times\mathcal{F}/\mathscr{B}(\mathbb{A}^{4}(\mathbb{V_{\mathrm{M}}^{\mathrm{4}}},g))$-measurable
$\hat{x}(\circ,\bullet)$ is a 4-dimentional (R0)-process if each
of $\{\varphi_{\mathbb{A}^{4}(\mathbb{V_{\mathrm{M}}^{\mathrm{4}}},g)}^{\mu}\circ\hat{x}(\tau,\bullet)\}_{\mu=0,1,2,3}$
belongs to $L^{1}(\varOmega,\mathscr{P})$ and the mapping $\tau\mapsto\hat{x}(\tau,\omega)$
is almost surely continuous.\end{leftbar}
\end{defn}
By employing {\bf Definition \ref {P_and_F}} and $\mathbb{L}_{T}^{p}(E)$
a family of $\mathscr{B}(\mathbb{R})\times\mathcal{F}/\mathscr{B}(E)$-measurable
mappings for a topological space $E$, let us introduce $\mathcal{L}_{\mathrm{loc}}^{p}(\{\mathscr{P}_{\tau}\};E)$
and $\mathcal{L}_{\mathrm{loc}}^{p}(\{\mathscr{F}_{\tau}\};E)$ as
families of stochastic processes as follows:

\[
\mathbb{L}_{T}^{p}(E)\coloneqq\left\{ \hat{X}(\circ,\bullet)\left|\hat{X}(\circ,\bullet):\mathbb{R}\times\varOmega\rightarrow E,\,\sum_{A}\int_{T\subset\mathbb{R}}|\varphi_{E}^{A}\circ\hat{X}(\tau,\omega)|^{p}d\tau<\infty\,\,\mathrm{a.s.}\right.\right\} 
\]
\[
\mathcal{L}_{\mathrm{loc}}^{p}(\{\mathscr{P}_{\tau}\};E)\coloneqq\left\{ \left.\hat{X}(\circ,\bullet)\in\mathbb{L}_{(-\infty,\tau]}^{p}(E)\right|\begin{gathered}\hat{X}(\circ,\bullet)\mathrm{\,is\,}\{\mathscr{P}_{\tau}\}\mathchar`-\mathrm{adapted}\end{gathered}
\right\} 
\]
\[
\mathcal{L}_{\mathrm{loc}}^{p}(\{\mathscr{F}_{\tau}\};E)\coloneqq\left\{ \left.\hat{X}(\circ,\bullet)\in\mathbb{L}_{[\tau,\infty)}^{p}(E)\right|\forall\tau\in\mathbb{R},\,\hat{X}^{\mu}(\circ,\bullet)\mathrm{\,is\,}\{\mathscr{F}_{\tau}\}\mathchar`-\mathrm{adapted}\right\} 
\]
For the later discussion, $\mathring{\epsilon}$ is defined as $\mathring{\epsilon}=1$
when a 1-dimensional $\hat{x}^{\mu}(\circ,\bullet)=\varphi_{\mathbb{A}^{4}(\mathbb{V_{\mathrm{M}}^{\mathrm{4}}},g)}^{\mu}\circ\hat{x}(\circ,\bullet)$
is $\{\mathcal{P}_{\tau}\}$-adapted and $\mathring{\epsilon}=-1$
if $\hat{x}^{\mu}(\circ,\bullet)$ is $\{\mathcal{F}_{\tau}\}$-adapted.
\begin{leftbar}
\begin{defn}[(R1)-process]
If $\hat{x}(\circ,\bullet)$ is an (R0)-process and a following $\mathcal{V}_{+}(\hat{x}(\circ,\bullet))\in\mathcal{L}_{\mathrm{loc}}^{1}(\{\mathscr{P}_{\tau}\};\mathbb{V}_{\mathrm{M}}^{\mathrm{4}})$
exists, $\hat{x}(\circ,\bullet)$ is named an (R1)-process.
\begin{equation}
\mathcal{V}_{+}(\hat{x}(\tau,\omega))\coloneqq\underset{\delta t\rightarrow0+}{\lim}\mathbb{E}\left\llbracket \left.\frac{\hat{x}(\tau+\mathring{\epsilon}\times\delta\tau,\bullet)-\hat{x}(\tau,\bullet)}{\mathring{\epsilon}\times\delta\tau}\right|\mathcal{\mathscr{P}}_{\tau}\right\rrbracket (\omega)
\end{equation}
 The definitions of its each components are the following:
\begin{equation}
\left\{ \begin{array}{c}
\begin{gathered}\mathcal{V}_{+}^{0}(\hat{x}(\tau,\omega))=\underset{\delta t\rightarrow0+}{\lim}\mathbb{E}\left\llbracket \left.\frac{\hat{x}^{0}(\tau,\bullet)-\hat{x}^{0}(\tau-\delta\tau,\bullet)}{\delta\tau}\right|\mathcal{F}_{\tau}\right\rrbracket (\omega)\end{gathered}
\\
\begin{gathered}\mathcal{V}_{+}^{i=1,2,3}(\hat{x}(\tau,\omega))=\underset{\delta t\rightarrow0+}{\lim}\mathbb{E}\left\llbracket \left.\frac{\hat{x}^{i}(\tau+\delta\tau,\bullet)-\hat{x}^{i}(\tau,\bullet)}{\delta\tau}\right|\mathcal{P}_{\tau}\right\rrbracket (\omega)\end{gathered}
\end{array}\right.
\end{equation}

\end{defn}
\end{leftbar}\begin{leftbar}
\begin{defn}[(S1)-process]
If $\hat{x}(\circ,\bullet)$ is an (R1)-process and a following $\mathcal{V}_{-}(\hat{x}(\circ,\bullet))\in\mathcal{L}_{\mathrm{loc}}^{1}(\{\mathscr{F}_{\tau}\};\mathbb{V}_{\mathrm{M}}^{\mathrm{4}})$
exists, $\hat{x}(\circ,\bullet)$ is named an (S1)-process. 
\begin{eqnarray}
\mathcal{V}_{-}(\hat{x}(\tau,\omega)) & \coloneqq & \underset{\delta t\rightarrow0+}{\lim}\mathbb{E}\left\llbracket \left.\frac{\hat{x}(\tau+\mathring{\epsilon}\times\delta\tau,\bullet)-\hat{x}(\tau,\bullet)}{\mathring{\epsilon}\times\delta\tau}\right|\mathscr{F_{\tau}}\right\rrbracket (\omega)
\end{eqnarray}

\end{defn}
\end{leftbar}Then, an (S1)-process provides us the relation of the
stochastic integral on $\tau_{a}\leq\tau\leq\tau_{b}$:
\begin{eqnarray}
\hat{x}^{\mu}(\tau,\omega) & = & \hat{x}^{\mu}(\tau_{a},\omega)+\int_{\tau_{a}}^{\tau}d\tau'\,\mathcal{V}_{+}^{\mu}(\hat{x}(\tau',\omega))+\int_{\tau_{a}}^{\tau}dy_{+}^{\mu}(\tau',\omega)\label{eq: S1-a}\\
 & = & \hat{x}^{\mu}(\tau_{b},\omega)-\int_{\tau}^{\tau_{b}}d\tau'\,\mathcal{V}_{-}^{\mu}(\hat{x}(\tau',\omega))-\int_{\tau}^{\tau_{b}}dy_{-}^{\mu}(\tau',\omega)\label{eq: S1-b}
\end{eqnarray}
Where, $y_{+}(\circ,\bullet)$ and $y_{-}(\circ,\bullet)$ of martingales
are $\{\mathcal{\mathscr{P}}_{\tau}\}$-adapted and $\{\mathcal{\mathscr{F}}_{\tau}\}$-adapted,
satisfy the following basic relations:
\begin{equation}
\mathbb{E}\left\llbracket \left.y_{+}(\tau+\mathring{\epsilon}\times\delta\tau,\bullet)-y_{+}(\tau,\bullet)\right|\mathcal{\mathscr{P}}_{\tau}\right\rrbracket (\omega)=0
\end{equation}
\begin{equation}
\mathbb{E}\left\llbracket \left.y_{-}(\tau+\mathring{\epsilon}\times\delta\tau,\bullet)-y_{-}(\tau,\bullet)\right|\mathcal{\mathscr{F}}_{\tau}\right\rrbracket (\omega)=0
\end{equation}

\begin{defn}[(R2)-process]
\begin{leftbar}When $\hat{x}(\circ,\bullet)$ is an (R1)-process
and let $y_{+}(\circ,\bullet)$ be $y_{+}(\tau+\mathring{\epsilon}\times\delta\tau,\bullet)-y_{+}(\tau,\bullet)\in\mathcal{L}_{\mathrm{loc}}^{2}(\{\mathscr{P}_{\tau}\};\mathbb{V}_{\mathrm{M}}^{\mathrm{4}})$,
$\hat{x}(\circ,\bullet)$ is named an (R2)-process if 
\begin{equation}
\mathbb{E}\left\llbracket \left.y_{+}(\tau+\mathring{\epsilon}\times\delta\tau,\bullet)-y_{+}(\tau,\bullet)\right|\mathcal{\mathscr{P}}_{\tau}\right\rrbracket (\omega)=0
\end{equation}
and a following $\sigma_{+}^{2}(\tau,\bullet)\in\mathcal{L}_{\mathrm{loc}}^{1}(\{\mathscr{P}_{\tau}\};\mathbb{V}_{\mathrm{M}}^{\mathrm{4}}\otimes\mathbb{V}_{\mathrm{M}}^{\mathrm{4}})$
exists, such that $\tau\mapsto\sigma_{+}^{2}(\tau,\omega)$ is continuous:
\begin{equation}
\sigma_{+}^{2}(\tau,\omega)\coloneqq\underset{\delta t\rightarrow0+}{\lim}\mathbb{E}\left\llbracket \left.\frac{[y_{+}(\tau+\mathring{\epsilon}\times\delta\tau,\bullet)-y_{+}(\tau,\bullet)]\otimes[y_{+}(\tau+\mathring{\epsilon}\times\delta\tau,\bullet)-y_{+}(\tau,\bullet)]}{\delta\tau}\right|\mathcal{\mathscr{P}}_{\tau}\right\rrbracket (\omega)
\end{equation}
\end{leftbar}
\end{defn}
\begin{leftbar}
\begin{defn}[(S2)-process]
When $\hat{x}(\circ,\bullet)$ is an (R2) and (S1)-process and let
$y_{-}(\circ,\bullet)$ be $y_{-}(\tau+\mathring{\epsilon}\times\delta\tau,\bullet)-y_{-}(\tau,\bullet)\in\mathcal{L}_{\mathrm{loc}}^{2}(\{\mathscr{F}_{\tau}\};\mathbb{V}_{\mathrm{M}}^{\mathrm{4}})$,
$\hat{x}(\circ,\bullet)$ is called an (S2)-process if 
\begin{equation}
\mathbb{E}\left\llbracket \left.y_{-}(\tau+\mathring{\epsilon}\times\delta\tau,\bullet)-y_{-}(\tau,\bullet)\right|\mathcal{\mathscr{F}}_{\tau}\right\rrbracket (\omega)=0\,,
\end{equation}
and a following $\sigma_{-}^{2}(\tau,\bullet)\in\mathcal{L}_{\mathrm{loc}}^{1}(\{\mathscr{F}_{\tau}\};\mathbb{V}_{\mathrm{M}}^{\mathrm{4}}\otimes\mathbb{V}_{\mathrm{M}}^{\mathrm{4}})$
exists, such that $\tau\mapsto\sigma_{-}^{2}(\tau,\omega)$ is continuous:
\begin{equation}
\sigma_{-}^{2}(\tau,\omega)\coloneqq\underset{\delta t\rightarrow0+}{\lim}\mathbb{E}\left\llbracket \left.\frac{[y_{-}(\tau+\mathring{\epsilon}\times\delta\tau,\bullet)-y_{-}(\tau,\bullet)]\otimes[y_{-}(\tau+\mathring{\epsilon}\times\delta\tau,\bullet)-y_{-}(\tau,\bullet)]}{\delta\tau}\right|\mathcal{\mathscr{F}}_{\tau}\right\rrbracket (\omega)
\end{equation}

\end{defn}
\end{leftbar}\begin{leftbar}
\begin{defn}[(R3)-process]
If $\hat{x}(\circ,\bullet)$ is an (R2)-process and $\det\sigma_{+}^{2}(\tau,\omega)>0$
is almost surely satisfied for each $t\in\mathbb{R}$, then, $\hat{x}(\circ,\bullet)$
is named an (R3)-process.
\end{defn}
\end{leftbar}\begin{leftbar}
\begin{defn}[(S3)-process]
If $\hat{x}(\circ,\bullet)$ is an (R3) and (S2)-process, $\det\sigma_{-}^{2}(\tau,\omega)>0$
is almost surely satisfied for each $t\in\mathbb{R}$, then, $\hat{x}(\circ,\bullet)$
is called an(S3)-process.
\end{defn}
\end{leftbar}

The discussion up to here is the simple extension from the original
idea by Nelson to a 4-dimensional composited process of $\{\mathscr{P}_{\tau}\}$
and $\{\mathscr{F}_{\tau}\}$. Where, $y_{\pm}(\tau,\omega)\coloneqq\lambda\times W_{\pm}(\tau,\omega)$
for $\lambda>0$ satisfies the above (S3) processes \cite{Nelson(2001_book),Nelson(1985_book)}.
Let $\{\mathcal{\mathscr{P}}_{\tau}\}$-progressive be a $\mathscr{B}((-\infty,\tau])\times\mathcal{\mathscr{P}}_{\tau}/\mathscr{B}(X)$-measurable
process and $\{\mathcal{\mathscr{F}}_{\tau}\}$-progressive be a $\mathscr{B}([\tau,\infty))\times\mathcal{\mathscr{F}}_{\tau}/\mathscr{B}(X)$-measurable
process for all $\tau\in\mathbb{R}$.
\begin{defn}[D-progressive $\hat{x}(\circ,\bullet)$]
\begin{leftbar}\label{D-progressive}A 4-dimentional (S3)-process
$\hat{x}(\circ,\bullet)$ on $\left(\mathit{\Omega},\mathcal{F},\mathscr{P}\right)$
is named ``the dual-progressively measurable process'', or by shortening
``D-progressive'' and also ``the D-process'' when $y_{\pm}(\circ,\bullet)\coloneqq\lambda\times W_{\pm}(\circ,\bullet)$
with respect to $\lambda>0$. For  $\tau_{a}\leq\tau\leq\tau_{b}$,
$\hat{x}(\circ,\bullet)$ is expressed by the following $\{\mathcal{\mathscr{P}}_{\tau}\}$-progressive
and $\{\mathcal{\mathscr{F}}_{\tau}\}$-progressive process: 
\begin{eqnarray}
\hat{x}^{\mu}(\tau,\omega) & = & \hat{x}^{\mu}(\tau_{a},\omega)+\int_{\tau_{a}}^{\tau}d\tau'\,\mathcal{V}_{+}^{\mu}(\hat{x}(\tau',\omega))+\lambda\times\int_{\tau_{a}}^{\tau}dW_{+}^{\mu}(\tau',\omega)\label{eq: Ito-path1}\\
 & = & \hat{x}^{\mu}(\tau_{b},\omega)-\int_{\tau}^{\tau_{b}}d\tau'\,\mathcal{V}_{-}^{\mu}(\hat{x}(\tau',\omega))-\lambda\times\int_{\tau}^{\tau_{b}}dW_{-}^{\mu}(\tau',\omega)\label{eq: Ito-path2}
\end{eqnarray}
These (\ref{eq: Ito-path1}-\ref{eq: Ito-path2}) is regarded as the
solution of the following stochastic differential equation:
\begin{equation}
\boxed{\ensuremath{d\hat{x}^{\mu}(\tau,\omega)=\mathcal{V}_{\pm}^{\mu}(\hat{x}(\tau,\omega))d\tau+\lambda\times dW_{\pm}^{\mu}(\tau,\omega)}}\label{eq:2-1}
\end{equation}
\end{leftbar} 
\end{defn}
For $\delta\tau>0$, let $d_{\pm}\hat{x}^{\mu}(\tau,\omega)$ be defined
by $\int_{\tau}^{\tau+\epsilon\mathring{\epsilon}\times\delta\tau}d_{\epsilon}\hat{x}^{\mu}(\tau',\omega)\coloneqq\hat{x}^{\mu}(\tau+\epsilon\mathring{\epsilon}\times\delta\tau,\omega)-\hat{x}^{\mu}(\tau,\omega)$
with the signature $\epsilon=\pm$. Since D-progressive $\hat{x}(\tau,\bullet)$
is $\mathcal{\mathscr{P}}_{\tau}/\mathscr{B}(\mathbb{A}^{4}(\mathbb{V_{\mathrm{M}}^{\mathrm{4}}},g))$
and $\mathcal{\mathscr{F}}_{\tau}/\mathscr{B}(\mathbb{A}^{4}(\mathbb{V_{\mathrm{M}}^{\mathrm{4}}},g))$-measurable
for all $\tau$, the following theorem is imposed:
\begin{thm}
\begin{leftbar}\label{D-progressive_2}A D-progressive $\hat{x}(\circ,\bullet)$
is $\{\mathcal{\mathscr{P}}_{\tau}\cap\mathcal{\mathscr{F}}_{\tau}\}$-adapted.\end{leftbar} 
\end{thm}
By using this {\bf Definition \ref{D-progressive}}, let us propose
the following conjecture.
\begin{conjecture}
\begin{leftbar}\label{conj_kinematics}A D-progressive $\hat{x}(\circ,\bullet)$
is a trajectory of a scalar electron satisfying the Klein-Gordon equation
when $\lambda=\sqrt{\hbar/m_{0}}$.\end{leftbar} 
\end{conjecture}
The demonstration of {\bf Conjecture \ref{conj_kinematics}} is the
center of our main issue in this paper and its feasibility is shown
in {\bf Section \ref{Sect3}}. 

Let $\hat{\xi}_{\pm}(\tau,\omega)$ be the white noise as the time
derivatives of $W_{\pm}(\tau,\omega)$ in means of the generalized-function
satisfying $\int_{\mathbb{R}}d\tau\,d\Phi/d\tau(\tau,\omega)\cdot W_{\pm}^{\mu}(\tau,\omega)=-\int_{\mathbb{R}}d\tau\,\Phi(\tau,\omega)\cdot\hat{\xi}_{\pm}^{\mu}(\tau,\omega)$
with respect to a test function $\Phi$ for all $\omega$. By introducing
$d_{\pm}\hat{x}(\tau,\omega)$ as the RHS in (\ref{eq:2-1}), (\ref{eq:2-1})
is recognized as the summation of the drift velocity $\mathcal{V}_{\pm}(\hat{x}(\tau,\omega))$
and the randomness $\lambda\times\hat{\xi}_{\pm}^{\mu}(\tau,\omega)=\lambda\times dW_{\pm}^{\mu}/d\tau(\tau,\omega)$,
\begin{equation}
\frac{d_{\pm}\hat{x}^{\mu}}{d\tau}(\tau,\omega)=\mathcal{V}_{\pm}^{\mu}(\hat{x}(\tau,\omega))+\lambda\times\hat{\xi}_{\pm}^{\mu}(\tau,\omega)\,.\label{eq:2-2}
\end{equation}
Since $\mathbb{E}\llbracket\hat{\xi}_{+}^{\mu}(\tau,\bullet)|\mathcal{\mathscr{P}}_{\tau}\rrbracket=0$
and $\mathbb{E}\llbracket\hat{\xi}_{-}^{\mu}(\tau,\bullet)|\mathcal{\mathscr{F}}_{\tau}\rrbracket=0$,
the conditional expectation of (\ref{eq:2-2}) (the mean-derivative)
imposes $\mathscr{\mathcal{V}}_{\pm}$ its drift velocity:
\begin{equation}
\begin{array}{cc}
\begin{gathered}\mathbb{E}\left\llbracket \left.\frac{d_{+}\hat{x}^{\mu}}{d\tau}(\tau,\bullet)\right|\mathcal{\mathscr{P}}_{\tau}\right\rrbracket (\omega)=\mathcal{V}_{+}^{\mu}(\hat{x}(\tau,\omega))\,,\end{gathered}
 & \begin{gathered}\mathbb{E}\left\llbracket \left.\frac{d_{-}\hat{x}^{\mu}}{d\tau}(\tau,\bullet)\right|\mathscr{F_{\tau}}\right\rrbracket (\omega)=\mathcal{V}_{-}^{\mu}(\hat{x}(\tau,\omega))\end{gathered}
\end{array}
\end{equation}
 In general, a D-progressive $\hat{x}(\circ,\bullet)$ imposes the
following It\^{o} formula \cite{Arai(2005),Ito(1944)}.
\begin{lem}[It\^{o} formula]
\begin{leftbar}\label{Ito formula}Consider a $C^{2}$-function
$f:\mathbb{A}^{4}(\mathbb{V_{\mathrm{M}}^{\mathrm{4}}},g)\rightarrow\mathbb{C}$,
the following It\^{o} formula with respect to a D-progressive $\hat{x}(\circ,\bullet)$
is found; 
\begin{equation}
d_{\pm}f(\hat{x}(\tau,\omega))=\partial_{\mu}f(\hat{x}(\tau,\omega))d_{\pm}\hat{x}^{\mu}(\tau,\omega)\mp\frac{\lambda^{2}}{2}\partial_{\mu}\partial^{\mu}f(\hat{x}(\tau,\omega))d\tau\,\,\mathrm{a.s.}\label{eq:2-Ito-formula}
\end{equation}
This is given by the following stochastic integral, too:
\begin{eqnarray}
f(\hat{x}(\tau_{b},\omega))-f(\hat{x}(\tau_{a},\omega)) & = & \int_{\tau_{a}}^{\tau_{b}}d_{\pm}f(\hat{x}(\tau,\omega))\\
 & = & \int_{\tau_{a}}^{\tau_{b}}d_{\pm}\hat{x}^{\mu}(\tau,\omega)\,\partial_{\mu}f(\hat{x}(\tau,\omega))\mp\frac{\lambda^{2}}{2}\int_{\tau_{a}}^{\tau_{b}}d\tau\,\partial_{\mu}\partial^{\mu}f(\hat{x}(\tau,\omega))\,\,\mathrm{a.s.}\label{eq: Ito-integral}
\end{eqnarray}
\end{leftbar}
\end{lem}

\subsection{Complex velocity}

In order to {\bf Theorem \ref{D-progressive_2}}, a D-progressively
measurable process is $\{\mathcal{\mathscr{P}}_{\tau}\cap\mathcal{\mathscr{F}}_{\tau}\}$-adapted.
Therefore, the superposition of $d_{+}$ and $d_{-}$ is introduced.
L. Nottale introduces the following complex differential $\hat{d}$
and the complex velocity $\mathscr{\mathcal{V}}(\hat{x}(\circ,\bullet))$
as the essential manners of quantum dynamics \cite{Nottale(2011)}.

\begin{defn}[Complex differential and velocity]
\begin{leftbar}Consider a $C^{2}$-function $f:\mathbb{A}^{4}(\mathbb{V_{\mathrm{M}}^{\mathrm{4}}},g)\rightarrow\mathbb{C}$
 and its It\^{o} formula $d_{\pm}f$ characterized by {\bf Lemma \ref{Ito formula}}.
Let $\hat{d}$ be the complex differential on a given D-progressive
$\hat{x}(\circ,\bullet)$: 
\begin{equation}
\hat{d}\coloneqq\frac{1-i}{2}d_{+}+\frac{1+i}{2}d_{-}\label{eq:2-3}
\end{equation}
\begin{equation}
\hat{d}f(\hat{x}(\tau,\omega))=\partial_{\mu}f(\hat{x}(\tau,\omega))\hat{d}\hat{x}^{\mu}(\tau,\omega)+\frac{i\lambda^{2}}{2}\partial^{\mu}\partial_{\mu}f(\hat{x}(\tau,\omega))d\tau\,\mathrm{a.s.}
\end{equation}
Then, consider a conditional expectation of the derivative given $\gamma_{\tau}\coloneqq\mathscr{P}_{\tau}\cap\mathscr{F}_{\tau}\subset\mathcal{F}$
is denoted by
\begin{equation}
\mathbb{E}\left\llbracket \left.\frac{\hat{d}f}{d\tau}(\hat{x}^{\mu}(\tau,\bullet))\right|\gamma_{\tau}\right\rrbracket (\omega)=\mathscr{\mathcal{V}}^{\mu}(\hat{x}(\tau,\omega))\partial_{\mu}f(\hat{x}(\tau,\omega))+\frac{i\lambda^{2}}{2}\partial^{\mu}\partial_{\mu}f(\hat{x}(\tau,\omega))\,,
\end{equation}
especially when $f(\hat{x}(\tau,\omega))=\hat{x}(\tau,\omega)$, it
derives the complex velocity $\mathscr{\mathcal{V}}\in\mathbb{V_{\mathrm{M}}^{\mathrm{4}}}\oplus i\mathbb{V_{\mathrm{M}}^{\mathrm{4}}}$,
\begin{equation}
\mathscr{\mathcal{V}}^{\mu}(\hat{x}(\tau,\omega))\coloneqq\mathbb{E}\left\llbracket \left.\frac{\hat{d}\hat{x}^{\mu}}{d\tau}(\tau,\bullet)\right|\gamma_{\tau}\right\rrbracket (\omega)=\frac{1-i}{2}\mathcal{V}_{+}^{\mu}(\hat{x}(\tau,\omega))+\frac{1+i}{2}\mathcal{V}_{-}^{\mu}(\hat{x}(\tau,\omega))\,.\label{eq: complex V}
\end{equation}
By choosing a $C^{2}$-function $\phi:\mathbb{A}^{4}(\mathbb{V_{\mathrm{M}}^{\mathrm{4}}},g)\rightarrow\mathbb{C}$
like Ref.\cite{Nottale(2011)}, the following is introduced: 
\begin{equation}
\mathcal{V}^{\alpha}(x)\coloneqq i\lambda^{2}\times\partial{}^{\alpha}\ln\phi(x)+\frac{e}{m_{0}}A{}^{\alpha}(x),\,\,x\in\hat{x}(\tau,\varOmega)\mathrm{\,\,for\,each\,}\tau\label{eq: Complex velocity}
\end{equation}
\end{leftbar}
\end{defn}
The the gauge invariance of scalar QED is found easily.
\begin{thm}[Gauge invariance of $\mathcal{V}$]
\begin{leftbar}\label{Gauge_inv_V}For a given $C^{1}$-function
$\varLambda:\mathbb{A}^{4}(\mathbb{V_{\mathrm{M}}^{\mathrm{4}}},g)\rightarrow\mathbb{R}$,
the complex velocity $\mathcal{V}^{\alpha}(x)$ satisfies the local
$U(1)$-gauge symmetry in the transformation $(\phi,A)\mapsto(\phi',A')$:
\begin{equation}
\begin{array}{ccc}
\phi'(x)=e^{-ie\varLambda(x)/\hbar}\times\phi(x)\,, &  & A'{}^{\alpha}(x)=A{}^{\alpha}(x)-\partial^{\alpha}\varLambda(x)\end{array}
\end{equation}
\end{leftbar}
\end{thm}

\subsection{Fokker-Planck equations}

 Let us consider a $C^{2}$-function $f:\mathbb{A}^{4}(\mathbb{V_{\mathrm{M}}^{\mathrm{4}}},g)\rightarrow\mathbb{C}$
on $(\mathbb{A}^{4}(\mathbb{V_{\mathrm{M}}^{\mathrm{4}}},g),\mathscr{B}(\mathbb{A}^{4}(\mathbb{V_{\mathrm{M}}^{\mathrm{4}}},g)),\mu)$
and its expectation $\mathbb{E}\llbracket f(\hat{x}(\tau,\bullet))\rrbracket$
at $\tau$. Where, $\mu:\mathbb{A}^{4}(\mathbb{V_{\mathrm{M}}^{\mathrm{4}}},g)\rightarrow[0,\infty)$.
The probability density $p:\mathbb{A}^{4}(\mathbb{V_{\mathrm{M}}^{\mathrm{4}}},g)\times\mathbb{R}\rightarrow[0,\infty)$
with respect to $\hat{x}(\circ,\bullet)$ is introduced by a the following
relation at $\tau\in\mathbb{R}$: 
\begin{equation}
\mathscr{P}(\varOmega)\coloneqq\int_{\hat{x}(\tau,\varOmega)\subset\mathbb{A}^{4}(\mathbb{V_{\mathrm{M}}^{\mathrm{4}}},g)}d\mu(x)\,p(x,\tau)=1\label{eq: P-measure}
\end{equation}
Where, $\hat{x}(\tau,\varOmega)\coloneqq\mathrm{supp}(p(\circ,\tau))$
and it leads $p(\mathbb{A}^{4}(\mathbb{V_{\mathrm{M}}^{\mathrm{4}}},g)\backslash\hat{x}(\tau,\varOmega),\tau)=0$.
Since 
\begin{eqnarray}
\mathbb{E}\llbracket f(\hat{x}(\tau,\bullet))\rrbracket & \coloneqq & \int_{\mathit{\Omega}}d\mathscr{P}(\omega)\,f(\hat{x}(\tau,\omega))\nonumber \\
 & = & \int_{\mathbb{A}^{4}(\mathbb{V_{\mathrm{M}}^{\mathrm{4}}},g)}d\mu(x)\,f(x)p(x,\tau)\,,
\end{eqnarray}
the probability density 
\begin{equation}
p(x,\tau)=\mathbb{E}\left\llbracket \delta^{4}(x-\hat{x}(\tau,\bullet))\right\rrbracket 
\end{equation}
is regarded as the kernel of a linear functional $\{\mathbb{E}\llbracket f(\hat{x}(\tau,\bullet))\rrbracket\}_{\forall f}$.
Consider the derivative of it with respect to $\tau$,
\begin{equation}
\frac{d}{d\tau}\mathbb{E}\llbracket f(\hat{x}(\tau,\bullet))\rrbracket=\int_{\mathbb{A}^{4}(\mathbb{V_{\mathrm{M}}^{\mathrm{4}}},g)}d\mu(x)\,f(x)\partial_{\tau}p(x,\tau)\,.\label{eq:d/dt E=00005Bf=00005D}
\end{equation}
The LHS of this equation (\ref{eq:d/dt E=00005Bf=00005D}) along the
evolution $d_{\pm}\hat{x}(\tau,\omega)$ is considered as follows;
\begin{eqnarray}
\frac{d}{d\tau}\mathbb{E}\llbracket f(\hat{x}(\tau,\bullet))\rrbracket & = & \mathbb{E}\left\llbracket \mathcal{V}_{\pm}^{\mu}(\hat{x}(\tau,\bullet))\partial_{\mu}f(\hat{x}(\tau,\bullet))\mp\frac{\lambda^{2}}{2}\partial^{\mu}\partial_{\mu}f(\hat{x}(\tau,\bullet))\right\rrbracket \nonumber \\
 & = & \int_{\mathbb{A}^{4}(\mathbb{V_{\mathrm{M}}^{\mathrm{4}}},g)}d\mu(x)\,f(x)\left\{ -\partial_{\mu}[\mathcal{V}_{\pm}^{\mu}(x)p(x,\tau)]\mp\frac{\lambda^{2}}{2}\partial^{\mu}\partial_{\mu}p(x,\tau)\right\} \,.\label{eq:d/dt E=00005Bf=00005D 2}
\end{eqnarray}
For an arbitrary $C^{2}$-function $f$, the following Fokker-Planck
equations of a D-progressive $\hat{x}(\circ,\bullet)$ are derived.
\begin{thm}[Fokker-Planck equations]
\begin{leftbar}\label{FP eq}Consider a D-progressive $\hat{x}(\circ,\bullet)$
on $(\mathit{\Omega},\mathcal{F},\mathscr{P})$, there is the $C^{2,1}$-probability
density of $\hat{x}(\circ,\bullet)$ such that $p:\mathbb{A}^{4}(\mathbb{V_{\mathrm{M}}^{\mathrm{4}}},g)\times\mathbb{R}\rightarrow[0,\infty)$
satisfying the following Fokker-Planck equation: 
\begin{equation}
\partial_{\tau}p(x,\tau)+\partial_{\mu}[\mathcal{V}_{\pm}^{\mu}(x)p(x,\tau)]\pm\frac{\lambda^{2}}{2}\partial^{\mu}\partial_{\mu}p(x,\tau)=0\label{eq: Fokker-Planck}
\end{equation}
By using the definition of $\mathscr{\mathcal{V}}\in\mathbb{V_{\mathrm{M}}^{\mathrm{4}}}\oplus i\mathbb{V_{\mathrm{M}}^{\mathrm{4}}}$
(see (\ref{eq: complex V})), the superpositions of the ``$\pm$''-Fokker-Planck
equations are found:
\begin{equation}
\partial_{\tau}p(x,\tau)+\partial_{\mu}\left[\mathrm{Re}\{\mathcal{V}^{\mu}(x)\}p(x,\tau)\right]=0\label{eq: eq of continuity}
\end{equation}
\begin{equation}
\mathrm{Im}\{\mathcal{V}^{\mu}(x)\}=\begin{cases}
\begin{gathered}\frac{\lambda^{2}}{2}\times\partial^{\mu}\ln p(x,\tau)\end{gathered}
, & \begin{gathered}x\in\hat{x}(\tau,\varOmega)\end{gathered}
\\
\begin{gathered}\frac{\lambda^{2}}{2}\times\partial^{\mu}\ln\int_{\mathbb{R}}d\tau\,p(x,\tau)\end{gathered}
, & x\in\bigcup_{\tau\in\mathbb{R}}\hat{x}(\tau,\varOmega)
\end{cases}\label{eq: osmotic pressure1}
\end{equation}
Where, $\bigcup_{\tau\in\mathbb{R}}\hat{x}(\tau,\varOmega)=\mathrm{supp}(\int_{\mathbb{R}}d\tau\,p(\circ,\tau))\subset\mathbb{A}^{4}(\mathbb{V_{\mathrm{M}}^{\mathrm{4}}},g)$.\end{leftbar}
\end{thm}
Equation (\ref{eq: eq of continuity}) represents the equation of
continuity of the probability density $p$ in $4+1$ dimensional space.
Then, $\partial_{\mu}\left[\mathrm{Re}\{\mathcal{V}^{\mu}(x)\}\int_{\mathbb{R}}d\tau\,p(x,\tau)\right]=0$
is also derived by selecting a natural boundary condition of $p(x,\partial\mathbb{R})=0$.
 Equation (\ref{eq: osmotic pressure1}) is a mimic of the osmotic
pressure formula \cite{Nelson(1966a),Nelson(2001_book)}. Where, the
above definition of $\mathrm{Im}\{\mathcal{V}^{\mu}(x)\}$ is identified
by its domain $\hat{x}(\tau,\varOmega)$ or $\bigcup_{\tau\in\mathbb{R}}\hat{x}(\tau,\varOmega)$.

\subsection{Proper time}

One of the delicate problem in this paper is the definition of the
proper time of a stochastic quanta on $(\mathbb{A}^{4}(\mathbb{V_{\mathrm{M}}^{\mathrm{4}}},g),\mathscr{B}(\mathbb{A}^{4}(\mathbb{V_{\mathrm{M}}^{\mathrm{4}}},g)),\mu)$.
Since we want to consider the correspondence between a D-process and
a classical kinematics, the limit $\hbar\rightarrow0$ of the proper
time in the present model has to imply one in classical dynamics.
The well-known proper time in classical dynamics is
\begin{equation}
d\tau\mathrm{|_{classical}}=\frac{1}{c}\times\sqrt{dx_{\mu}(\tau)dx^{\mu}(\tau)}\,,\label{eq: dtau in classical}
\end{equation}
here, the metric $g=\mathrm{diag}(+1,-1,-1,-1)$ is selected. Let
us recall the following relation in advance; 
\begin{equation}
[\hat{d}^{*}\hat{x}_{\mu}(\tau,\omega)-\lambda\times\hat{d}^{*}W_{\mu}(\tau,\omega)]\cdot[\hat{d}\hat{x}^{\mu}(\tau,\omega)-\lambda\times\hat{d}W^{\mu}(\tau,\omega)]=\mathcal{V}_{\mu}^{*}(\hat{x}(\tau,\omega))\mathcal{V}^{\mu}(\hat{x}(\tau,\omega))d\tau^{2}.
\end{equation}
Where, $\hat{d}W(\tau,\omega)\coloneqq(1-i)/2\times dW_{+}(\tau,\omega)+(1+i)/2\times dW_{-}(\tau,\omega)$,
$\hat{d}^{*}W(\tau,\omega)\coloneqq[\hat{d}W(\tau,\omega)]^{*}$ and
$A^{*}$ represents the complex conjugate of $A$. Again, remind the
definition of the complex velocity
\begin{equation}
\mathcal{V}^{\mu}(x)=\frac{1}{m_{0}}\times\frac{i\hbar\partial{}^{\mu}\phi(x)+eA{}^{\mu}(x)\phi(x)}{\phi(x)}=\frac{1}{m_{0}}\times\frac{i\hbar\mathfrak{D}{}^{\mu}\phi(x)}{\phi(x)}\label{eq: 2-7}
\end{equation}
with respect to $x\in\hat{x}(\tau,\varOmega)$, $\mathcal{V}_{\mu}^{*}(x)\mathcal{V}^{\mu}(x)$
becomes 
\begin{eqnarray}
\mathcal{V}_{\mu}^{*}(x)\mathcal{V}^{\mu}(x) & = & \frac{1}{2m_{0}^{2}}\times\frac{\phi(x)(-i\hbar\mathfrak{D}_{\mu}^{*})\cdot(-i\hbar\mathfrak{D}^{*\mu})\phi^{*}(x)+\phi^{*}(x)(i\hbar\mathfrak{D}_{\mu})\cdot(i\hbar\mathfrak{D}^{\mu})\phi(x)}{\phi^{*}(x)\phi(x)}\nonumber \\
 &  & +\frac{\hbar^{2}}{2m_{0}^{2}}\times\frac{\partial_{\mu}\partial^{\mu}[\phi(x)\cdot\phi^{*}(x)]}{\phi^{*}(x)\phi(x)}\,.\label{eq: 2-8}
\end{eqnarray}
Let the $C^{2}$-function $\phi(x):\mathbb{A}^{4}(\mathbb{V_{\mathrm{M}}^{\mathrm{4}}},g)\rightarrow\mathbb{C}$
be the wave function of the complex Klein-Gordon equation, $(i\hbar\mathfrak{D}_{\mu})\cdot(i\hbar\mathfrak{D}^{\mu})\phi(x)-m_{0}^{2}c^{2}\phi(x)=0$.
 Due to this assumption, the first term in the RHS of (\ref{eq: 2-8})
is a constant of $c^{2}$. Then, the issue is the behavior of $\hbar^{2}/m_{0}^{2}\times\partial_{\mu}\partial^{\mu}[\phi^{*}(x)\phi(x)]/\phi^{*}(x)\phi(x)$.
Where, we follow the proposal by T. Zastawniak in Ref.\cite{Zastawniak(1990)}.
By defining the function $\phi(x)\coloneqq\exp[R(x)/\hbar+iS(x)/\hbar]$
with respect to real valued functions of $R$ and $S$, $\phi^{*}(x)\phi(x)=\exp[2R(x)/\hbar]$
is satisfied. Due to the definition of (\ref{eq: 2-7}), $\partial^{\mu}R(x)=\mathrm{Im}\{m_{0}\mathcal{V}^{\mu}(x)\}=\hbar/2\times\partial^{\mu}\ln p(x,\tau)$
on $x\in\hat{x}(\tau,\varOmega)$ (see (\ref{eq: osmotic pressure1}));
\begin{equation}
\frac{\hbar^{2}}{2m_{0}^{2}}\times\frac{\partial_{\mu}\partial^{\mu}[\phi(x)\cdot\phi^{*}(x)]}{\phi^{*}(x)\phi(x)}=\frac{\hbar^{2}}{2m_{0}^{2}}\times\frac{\partial_{\mu}\partial^{\mu}p(x,\tau)}{p(x,\tau)}\,.\label{eq: sub-eq-001}
\end{equation}
Hence, the second term of the RHS in (\ref{eq: 2-8}) is non-zero.
However, let us introduce the expectation of (\ref{eq: sub-eq-001})
after the substitution of $x=\hat{x}(\tau,\omega)$,
\begin{eqnarray}
\mathbb{E}\left\llbracket \frac{\hbar^{2}}{2m_{0}^{2}}\times\frac{\partial_{\mu}\partial^{\mu}p(\hat{x}(\tau,\bullet),\tau)}{p(\hat{x}(\tau,\bullet),\tau)}\right\rrbracket  & = & \frac{\hbar^{2}}{2m_{0}^{2}}\times\int_{\hat{x}(\tau,\varOmega)}d\mu(x)\,\left[\frac{\partial_{\mu}\partial^{\mu}p(x,\tau)}{p(x,\tau)}\right]p(x,\tau)\nonumber \\
 & = & \frac{\hbar^{2}}{2m_{0}^{2}}\times\int_{\hat{x}(\tau,\varOmega)}d\mu(x)\,\partial_{\mu}\partial^{\mu}p(x,\tau)=0,
\end{eqnarray}
by employing $\partial^{\mu}p(x,\tau)|_{x\in\partial\hat{x}(\tau,\varOmega)}=0$,
where $\partial\hat{x}(\tau,\varOmega)$ denotes the boundary of $\hat{x}(\tau,\varOmega)$.
Therefore the following relation is realized.
\begin{lem}[Lorentz invariant]
\begin{leftbar}\label{Lorentz_inv}Consider a D-progressive $\hat{x}(\circ,\bullet)$.
A stochastic kinematics of a scalar electron satisfies the following
Lorentz invariant for all $\tau\in\mathbb{R}$ \cite{Zastawniak(1990)}.
\begin{equation}
\boxed{\mathbb{E}\left\llbracket \mathcal{V}_{\mu}^{*}(\hat{x}(\tau,\bullet))\mathcal{V}^{\mu}(\hat{x}(\tau,\bullet))\right\rrbracket =c^{2}}\label{eq: VV}
\end{equation}
\end{leftbar}
\end{lem}
Due to this {\bf Lemma \ref{Lorentz_inv}}, the proper time is defined
as the mimic of classical dynamics (\ref{eq: dtau in classical}).
\begin{defn}[Proper time]
\begin{leftbar}\label{Proper_TIME}For all $\tau\in\mathbb{R}$,
the proper time for a stochastic kinematics  is defined by the following
invariant parameter;
\begin{equation}
\ensuremath{d\tau\coloneqq\frac{1}{c}\times\sqrt{\mathbb{E}\llbracket[\hat{d}^{*}\hat{x}_{\mu}(\tau,\bullet)-\lambda\times\hat{d}^{*}W_{\mu}(\tau,\bullet)]\cdot[\hat{d}\hat{x}^{\mu}(\tau,\bullet)-\lambda\times\hat{d}W^{\mu}(\tau,\bullet)]\rrbracket}}\,.\label{eq: proper time original}
\end{equation}
\end{leftbar}
\end{defn}

\section{Dynamics of a scalar electron and fields\label{Sect3}}

In order to the realization of the kinematics (\ref{eq:2-1}), i.e.,
$\ensuremath{d\hat{x}(\tau,\omega)=\mathcal{V}_{\pm}(\hat{x}(\tau,\omega))d\tau+\lambda\times dW_{\pm}(\tau,\omega)}$,
we need to investigate the behavior of the complex velocity $\mathscr{\mathcal{V}}^{\mu}(\hat{x}(\tau,\omega))\in\mathbb{V_{\mathrm{M}}^{\mathrm{4}}}\oplus i\mathbb{V_{\mathrm{M}}^{\mathrm{4}}}$.
For the derivation of $\mathcal{V}(\hat{x}(\tau,\omega))$, the action
integral (the functional) along a stochastic trajectory is required.
Before entering to the main body, we consider the variational calculus
associated with a stochastic particle briefly. After this explanation,
let us proceed the concrete definition of the action integral and
fields corresponding to the styles in classical dynamics.

\subsection{Euler-Lagrange (Yasue) equation}

In this small section, we focus the action integral on a stochastic
process. Concerning the complex velocity $\mathcal{V}\in\mathbb{V_{\mathrm{M}}^{\mathrm{4}}}\oplus i\mathbb{V_{\mathrm{M}}^{\mathrm{4}}}$,
L. Nottale suggests the following Lagrangian due to its forward and
backward evolution, i.e., $d_{\pm}\hat{x}(\circ,\bullet)$; $L_{0}(\tau,\hat{x},\mathcal{V}_{+},\mathcal{V}_{-})=L(\tau,\hat{x},\mathcal{V})$
\cite{Nottale(2011)}. However, we propose its extension, namely,
$L_{0}(\tau,\hat{x},\mathcal{V}_{+},\mathcal{V}_{-})=L(\tau,\hat{x},\mathcal{V},\mathcal{V^{\mathrm{*}}})$.
Here, $\mathcal{V^{\mathrm{*}}}\in\mathbb{V_{\mathrm{M}}^{\mathrm{4}}}\oplus i\mathbb{V_{\mathrm{M}}^{\mathrm{4}}}$
is the complex conjugate of $\mathcal{V}$. Recalling the definition
of $\mathcal{V}$ in (\ref{eq: complex V}), it is found 
\begin{equation}
L_{0}(\tau,\hat{x},\mathcal{V}_{+},\mathcal{V}_{-})=L\left(\tau,\hat{x},\frac{1-i}{2}\mathcal{V}_{+}+\frac{1+i}{2}\mathcal{V}_{-},\frac{1+i}{2}\mathcal{V}_{+}+\frac{1-i}{2}\mathcal{V}_{-}\right)
\end{equation}
on a sample path of a D-progressive $\hat{x}(\circ,\omega)$. For
the simplification, the following signatures are introduced ($\gamma_{\tau}\coloneqq\mathscr{P}_{\tau}\cap\mathscr{F}_{\tau}$
as the ``present'' $\tau$):
\begin{equation}
\mathfrak{D}_{\tau}^{\pm}\coloneqq\mathbb{E}\left\llbracket \left.\frac{d_{\pm}}{d\tau}\right|\gamma_{\tau}\right\rrbracket 
\end{equation}
\begin{equation}
\mathfrak{D}_{\tau}\coloneqq\mathbb{E}\left\llbracket \left.\frac{\hat{d}}{d\tau}\right|\gamma_{\tau}\right\rrbracket =\frac{1-i}{2}\mathfrak{D}_{\tau}^{+}+\frac{1+i}{2}\mathfrak{D}_{\tau}^{-}
\end{equation}
By using these expressions, $\mathfrak{D}_{\tau}^{\pm}\hat{x}^{\mu}(\tau,\omega)=\mathcal{V}_{\pm}^{\mu}(\hat{x}(\tau,\omega))$
is obviously satisfied (the {\bf mean derivatives}). The variation
of the functional $\int_{\tau_{1}}^{\tau_{2}}d\tau\,\mathbb{E}\left\llbracket L_{0}(\tau,\hat{x},\mathcal{V}_{+},\mathcal{V}_{-})\right\rrbracket $
with respect to $\hat{x}$ is
\begin{eqnarray}
\delta\int_{\tau_{1}}^{\tau_{2}}d\tau\,\mathbb{E}\left\llbracket L_{0}(\tau,\hat{x},\mathcal{V}_{+},\mathcal{V}_{-})\right\rrbracket  & = & \int_{\tau_{1}}^{\tau_{2}}d\tau\,\mathbb{E}\left\llbracket \frac{\partial L_{0}}{\partial\hat{x}^{\mu}}\delta\hat{x}^{\mu}+\frac{\partial L_{0}}{\partial\mathcal{V}_{+}^{\mu}}\delta\mathcal{V}_{+}^{\mu}+\frac{\partial L_{0}}{\partial\mathcal{V}_{-}^{\mu}}\delta\mathcal{V}_{-}^{\mu}\right\rrbracket \nonumber \\
 & = & \int_{\tau_{1}}^{\tau_{2}}d\tau\,\mathbb{E}\left\llbracket \frac{\partial L}{\partial\hat{x}^{\mu}}\delta\hat{x}^{\mu}+\frac{\partial L}{\partial\mathcal{V}^{\mu}}\mathfrak{D}_{\tau}\delta\hat{x}^{\mu}+\frac{\partial L}{\partial\mathcal{V}^{*\mu}}\mathfrak{D}_{\tau}^{*}\delta\hat{x}^{\mu}\right\rrbracket \,,\label{eq: variational 1}
\end{eqnarray}
where, the following relations are introduced.
\begin{equation}
\frac{\partial L_{0}}{\partial\hat{x}^{\mu}}=\frac{\partial L}{\partial\hat{x}^{\mu}}
\end{equation}
\begin{equation}
\frac{\partial L_{0}}{\partial\mathcal{V}_{+}^{\mu}}=\frac{1+i}{2}\frac{\partial L}{\partial\mathcal{V}^{\mu}}+\frac{1-i}{2}\frac{\partial L}{\partial\mathcal{V}^{*\mu}}
\end{equation}
\begin{equation}
\frac{\partial L_{0}}{\partial\mathcal{V}_{-}^{\mu}}=\frac{1-i}{2}\frac{\partial L}{\partial\mathcal{V}^{\mu}}+\frac{1+i}{2}\frac{\partial L}{\partial\mathcal{V}^{*\mu}}
\end{equation}
Then, we need to recall the following Nelson's partial integral \cite{Nelson(1966a),Nelson(2001_book),Yasue(1981a)}. 
\begin{lem}[Nelson's partial integral]
\begin{leftbar}\label{Nelson_partial}Let $\alpha,\beta:\mathbb{A}^{4}(\mathbb{V_{\mathrm{M}}^{\mathrm{4}}},g)\rightarrow\mathbb{V_{\mathrm{M}}^{\mathrm{4}}}\oplus i\mathbb{V_{\mathrm{M}}^{\mathrm{4}}}$
be a $C^{2}$-functions on a D-progressive $\hat{x}(\circ,\bullet)$,
the following partial integral formula is fulfilled;
\begin{multline}
\int_{\tau_{1}}^{\tau_{2}}d\tau\,\mathbb{E}\left\llbracket \mathfrak{D}_{\tau}^{\pm}\alpha_{\mu}(\hat{x}(\tau,\bullet))\cdot\beta^{\mu}(\hat{x}(\tau,\bullet))+\alpha_{\mu}(\hat{x}(\tau,\bullet))\cdot\mathfrak{D}_{\tau}^{\mp}\beta^{\mu}(\hat{x}(\tau,\bullet))\right\rrbracket \\
=\mathbb{E}\left\llbracket \alpha_{\mu}(\hat{x}(\tau_{2},\bullet))\beta^{\mu}(\hat{x}(\tau_{2},\bullet))-\alpha_{\mu}(\hat{x}(\tau_{1},\bullet))\beta^{\mu}(\hat{x}(\tau_{1},\bullet))\right\rrbracket \,,\label{eq:partial int formula original1}
\end{multline}
 or its differential form,
\begin{equation}
\frac{d}{d\tau}\mathbb{E}\left\llbracket \alpha_{\mu}(\hat{x}(\tau,\bullet))\beta^{\mu}(\hat{x}(\tau,\bullet))\right\rrbracket =\mathbb{E}\left\llbracket \begin{gathered}\mathfrak{D}_{\tau}^{\pm}\alpha_{\mu}(\hat{x}(\tau,\bullet))\cdot\beta^{\mu}(\hat{x}(\tau,\bullet))+\alpha_{\mu}(\hat{x}(\tau,\bullet))\cdot\mathfrak{D}_{\tau}^{\mp}\beta^{\mu}(\hat{x}(\tau,\bullet))\end{gathered}
\right\rrbracket \,.\label{eq:partial int formula original2}
\end{equation}
By using the superposition of the above ``$\pm$''-formulas, it
can be switched to the formula by the complex derivatives.
\begin{eqnarray}
\frac{d}{d\tau}\mathbb{E}\left\llbracket \alpha_{\mu}(\hat{x}(\tau,\bullet))\beta^{\mu}(\hat{x}(\tau,\bullet))\right\rrbracket  & = & \mathbb{E}\left\llbracket \mathfrak{D}_{\tau}\alpha_{\mu}(\hat{x}(\tau,\bullet))\cdot\beta^{\mu}(\hat{x}(\tau,\bullet))+\alpha_{\mu}(\hat{x}(\tau,\bullet))\cdot\mathfrak{D}_{\tau}^{*}\beta^{\mu}(\hat{x}(\tau,\bullet))\right\rrbracket \nonumber \\
 & = & \mathbb{E}\left\llbracket \mathfrak{D}_{\tau}^{*}\alpha_{\mu}(\hat{x}(\tau,\bullet))\cdot\beta^{\mu}(\hat{x}(\tau,\bullet))+\alpha_{\mu}(\hat{x}(\tau,\bullet))\cdot\mathfrak{D}_{\tau}\beta^{\mu}(\hat{x}(\tau,\bullet))\right\rrbracket \label{eq: Partial int formula}
\end{eqnarray}
\end{leftbar}\end{lem}
\begin{proof}
Consider the following relation at first:
\begin{multline}
\mathbb{E}\left\llbracket \mathfrak{D}_{\tau}^{+}\alpha_{\mu}(\hat{x}(\tau,\bullet))\cdot\beta^{\mu}(\hat{x}(\tau,\bullet))+\alpha_{\mu}(\hat{x}(\tau,\bullet))\cdot\mathfrak{D}_{\tau}^{-}\beta^{\mu}(\hat{x}(\tau,\bullet))\right\rrbracket \\
=\mathbb{E}\left\llbracket \mathfrak{D}_{\tau}^{-}\alpha_{\mu}(\hat{x}(\tau,\bullet))\cdot\beta^{\mu}(\hat{x}(\tau,\bullet))+\alpha_{\mu}(\hat{x}(\tau,\bullet))\cdot\mathfrak{D}_{\tau}^{+}\beta^{\mu}(\hat{x}(\tau,\bullet))\right\rrbracket \label{eq: D+D-=00003DD-D+}
\end{multline}
It is derived by using the It\^{o} formula of (\ref{eq:2-Ito-formula})
and (\ref{eq: osmotic pressure1}), 
\begin{multline}
\mathbb{E}\left\llbracket \mathfrak{D}_{\tau}^{+}\alpha_{\mu}(\hat{x}(\tau,\bullet))\cdot\beta^{\mu}(\hat{x}(\tau,\bullet))+\alpha_{\mu}(\hat{x}(\tau,\bullet))\cdot\mathfrak{D}_{\tau}^{-}\beta^{\mu}(\hat{x}(\tau,\bullet))\right\rrbracket \\
-\mathbb{E}\left\llbracket \mathfrak{D}_{\tau}^{-}\alpha_{\mu}(\hat{x}(\tau,\bullet))\cdot\beta^{\mu}(\hat{x}(\tau,\bullet))+\alpha_{\mu}(\hat{x}(\tau,\bullet))\cdot\mathfrak{D}_{\tau}^{+}\beta^{\mu}(\hat{x}(\tau,\bullet))\right\rrbracket \\
=-\lambda^{2}\times\int_{\mathbb{A}^{4}(\mathbb{V_{\mathrm{M}}^{\mathrm{4}}},g)}d\mu(x)\,\partial^{\nu}\left\{ p(x,\tau)\left[\begin{gathered}\partial_{\nu}\alpha_{\mu}(x)\cdot\beta^{\mu}(x)\\
-\alpha_{\mu}(x)\cdot\partial_{\nu}\beta^{\mu}(x)
\end{gathered}
\right]\right\} =0.
\end{multline}
Then, by recalling the Fokker-Planck equation (\ref{eq: Fokker-Planck}),
\begin{eqnarray}
\frac{d}{d\tau}\mathbb{E}\left\llbracket \alpha_{\mu}(\hat{x}(\tau,\bullet))\beta^{\mu}(\hat{x}(\tau,\bullet))\right\rrbracket  & = & \int_{\mathbb{A}^{4}(\mathbb{V_{\mathrm{M}}^{\mathrm{4}}},g)}d\mu(x)\,\alpha_{\mu}(x)\beta^{\mu}(x)\partial_{\tau}p(x,\tau)\nonumber \\
 & = & \frac{1}{2}\times\mathbb{E}\left\llbracket \begin{gathered}(\mathfrak{D}_{\tau}^{+}+\mathfrak{D}_{\tau}^{-})\alpha_{\mu}(\hat{x}(\tau,\bullet))\cdot\beta^{\mu}(\hat{x}(\tau,\bullet))\\
+\alpha_{\mu}(\hat{x}(\tau,\bullet))\cdot(\mathfrak{D}_{\tau}^{+}+\mathfrak{D}_{\tau}^{-})\beta^{\mu}(\hat{x}(\tau,\bullet))
\end{gathered}
\right\rrbracket \,,
\end{eqnarray}
equation (\ref{eq:partial int formula original2}) is demonstrated.
(\ref{eq: Partial int formula}) is also imposed by the superposition
of (\ref{eq:partial int formula original2}) for ``$\pm$''.
\end{proof}
By considering (\ref{eq: Partial int formula}), (\ref{eq: variational 1})
becomes
\begin{eqnarray}
\delta\int_{\tau_{1}}^{\tau_{2}}d\tau\,\mathbb{E}\left\llbracket L_{0}(\tau,\hat{x},\mathcal{V}_{+},\mathcal{V}_{-})\right\rrbracket  & = & \int_{\tau_{1}}^{\tau_{2}}d\tau\,\mathbb{E}\left\llbracket \left(\frac{\partial L}{\partial\hat{x}^{\mu}}-\mathfrak{D}_{\tau}^{*}\frac{\partial L}{\partial\mathcal{V}^{\mu}}-\mathfrak{D}_{\tau}\frac{\partial L}{\partial\mathcal{V}^{*\mu}}\right)\delta\hat{x}^{\mu}\right\rrbracket \nonumber \\
 &  & +\int_{\tau_{1}}^{\tau_{2}}d\tau\,\frac{d}{d\tau}\mathbb{E}\left\llbracket \frac{\partial L}{\partial\mathcal{V}^{\mu}}\delta\hat{x}^{\mu}+\frac{\partial L}{\partial\mathcal{V}^{*\mu}}\delta\hat{x}^{\mu}\right\rrbracket \,,
\end{eqnarray}
the following {\bf Theorem \ref{EL_eq}} is derived with the help
of $\delta\hat{x}^{\mu}(\tau_{i},\bullet)=0$ ($i=1,2$).
\begin{thm}[Euler-Lagrange (Yasue) equation]
\begin{leftbar}\label{EL_eq}Let the functional 
\begin{equation}
\mathfrak{S}[\hat{x},\mathcal{V},\mathcal{V^{\mathrm{*}}}]=\int_{\tau_{1}}^{\tau_{2}}d\tau\,\mathbb{E}\left\llbracket L\left(\tau,\hat{x}(\tau,\bullet),\mathcal{V}(\hat{x}(\tau,\bullet)),\mathcal{V^{\mathrm{*}}}(\hat{x}(\tau,\bullet))\right)\right\rrbracket 
\end{equation}
be the action integral on a D-progressive $\hat{x}(\circ,\bullet)$.
By its variation with respect to $\hat{x}(\circ,\bullet)$, the following
Euler-Lagrange (Yasue) equation is induced:
\begin{equation}
\boxed{\frac{\partial L}{\partial\hat{x}^{\mu}}-\mathfrak{D}_{\tau}^{*}\frac{\partial L}{\partial\mathcal{V}^{\mu}}-\mathfrak{D}_{\tau}\frac{\partial L}{\partial\mathcal{V}^{*\mu}}=0}\label{eq: EL eq}
\end{equation}
\end{leftbar}
\end{thm}

\subsection{Action integral}

Let us consider the action integral of ``classical'' dynamics on
$(\mathbb{A}^{4}(\mathbb{V_{\mathrm{M}}^{\mathrm{4}}},g),\mathscr{B}(\mathbb{A}^{4}(\mathbb{V_{\mathrm{M}}^{\mathrm{4}}},g)),\mu)$;
\begin{eqnarray}
S_{\mathrm{classical}} & = & \int_{\mathbb{R}}d\tau\,\frac{m_{0}}{2}v_{\alpha}(\tau)v^{\alpha}(\tau)-\int_{\mathbb{R}}d\tau\,eA_{\alpha}(x(\tau))v^{\alpha}(\tau)+\int_{\mathbb{A}^{4}(\mathbb{V_{\mathrm{M}}^{\mathrm{4}}},g)}d\mu(x)\,\frac{1}{4\mu_{0}c}F_{\alpha\beta}(x)F^{\alpha\beta}(x)\,.\label{eq:3-1}
\end{eqnarray}
Corresponding to (\ref{eq:3-1}), a new action integral of a stochastic
particle and a field is proposed via the introduction of the mass
measure and the charge measure: 
\begin{defn}[Mass and charge measures]
\begin{leftbar}Let $\mathscr{\mathcal{\mathfrak{M}}}$ and $\mathfrak{\mathfrak{E}}$
be the mass measure and the charge measure of a stochastic scalar
electron. For the positive constants $m_{0}$ and $e$, $\mathscr{\mathcal{\mathfrak{M}}}$
and $\mathfrak{\mathfrak{E}}$ are characterized by
\begin{equation}
\int_{\mathbb{A}^{4}(\mathbb{V_{\mathrm{M}}^{\mathrm{4}}},g)}d\mathcal{\mathfrak{M}}(x,\tau)\coloneqq m_{0}\times\int_{\mathbb{A}^{4}(\mathbb{V_{\mathrm{M}}^{\mathrm{4}}},g)}d\mu(x)\,\mathbb{E}\left\llbracket \delta^{4}(x-\hat{x}(\tau,\bullet))\right\rrbracket ,
\end{equation}
\begin{equation}
\int_{\mathbb{A}^{4}(\mathbb{V_{\mathrm{M}}^{\mathrm{4}}},g)}d\mathcal{\mathfrak{E}}(x,\tau)\coloneqq e\times\int_{\mathbb{A}^{4}(\mathbb{V_{\mathrm{M}}^{\mathrm{4}}},g)}d\mu(x)\,\mathbb{E}\left\llbracket \delta^{4}(x-\hat{x}(\tau,\bullet))\right\rrbracket .
\end{equation}
\end{leftbar}
\end{defn}
The key of this definition is the appearance of the smeared distribution
$\mathbb{E}\left\llbracket \delta^{4}(x-\hat{x}(\tau,\bullet))\right\rrbracket d\mu(x)$
from $\delta^{4}(x-x(\tau))d\mu(x)$ in classical dynamics.
\begin{defn}[Action integral]
\begin{leftbar}\label{Action integral}The following functional
$\mathfrak{S}$ is the action integral deriving the dynamics of a
``stochastic'' scalar electron and a field characterized by $\mathscr{\mathcal{V}}\in\mathbb{V_{\mathrm{M}}^{\mathrm{4}}}\oplus i\mathbb{V_{\mathrm{M}}^{\mathrm{4}}}$,
$A\in\mathbb{V_{\mathrm{M}}^{\mathrm{4}}}$ with the help by $F\in\mathbb{V_{\mathrm{M}}^{\mathrm{4}}}\otimes\mathbb{V_{\mathrm{M}}^{\mathrm{4}}}$
and a given tensor $\delta f\in\mathbb{V_{\mathrm{M}}^{\mathrm{4}}}\otimes\mathbb{V_{\mathrm{M}}^{\mathrm{4}}}$:
\begin{eqnarray}
\mathfrak{S}[\hat{x},\mathcal{V},\mathcal{V}^{*},A] & = & \int_{\mathbb{R}}d\tau\int_{\mathbb{A}^{4}(\mathbb{V_{\mathrm{M}}^{\mathrm{4}}},g)}d\mathcal{\mathfrak{M}}(x,\tau)\,\frac{1}{2}\mathcal{V}_{\alpha}^{*}(x)\mathcal{V}^{\alpha}(x)\nonumber \\
 &  & -\int_{\mathbb{R}}d\tau\int_{\mathbb{A}^{4}(\mathbb{V_{\mathrm{M}}^{\mathrm{4}}},g)}d\mathcal{\mathfrak{E}}(x,\tau)\,A_{\alpha}(x)\mathrm{Re}\left\{ \mathcal{V}^{\alpha}(x)\right\} \nonumber \\
 &  & +\int_{\mathbb{A}^{4}(\mathbb{V_{\mathrm{M}}^{\mathrm{4}}},g)}d\mu(x)\,\frac{1}{4\mu_{0}c}[F_{\alpha\beta}(x)+\delta f_{\alpha\beta}(x)]\cdot[F^{\alpha\beta}(x)+\delta f^{\alpha\beta}(x)]\label{eq:Lagrangian-1}
\end{eqnarray}
Where, $F^{\alpha\beta}(x)\coloneqq\partial^{\mu}A^{\nu}-\partial^{\nu}A^{\mu}$.
By writing the detail of the measures explicitly,
\begin{eqnarray}
\mathfrak{S}[\hat{x},\mathcal{V},\mathcal{V}^{*},A] & = & \mathbb{E}\left\llbracket \int_{\mathbb{R}}d\tau\,\frac{m_{0}}{2}\mathcal{V}_{\alpha}^{*}(\hat{x}(\tau,\bullet))\mathcal{V}^{\alpha}(\hat{x}(\tau,\bullet))\right\rrbracket \nonumber \\
 &  & +\mathbb{E}\left\llbracket -\int_{\mathbb{R}}d\tau\,e\,A_{\alpha}(\hat{x}(\tau,\bullet))\mathrm{Re}\{\mathcal{V}^{\alpha}(\hat{x}(\tau,\bullet))\}\right\rrbracket \nonumber \\
 &  & +\frac{1}{4\mu_{0}c}\int_{\mathbb{A}^{4}(\mathbb{V_{\mathrm{M}}^{\mathrm{4}}},g)}d\mu(x)\,[F_{\alpha\beta}(x)+\delta f_{\alpha\beta}(x)]\cdot[F^{\alpha\beta}(x)+\delta f^{\alpha\beta}(x)]\,.\label{eq:Lagrangian-1'}
\end{eqnarray}
\end{leftbar}
\end{defn}
Hence, the Lagrangian density $\mathfrak{L}$ is also introduced;
\begin{eqnarray}
\mathfrak{L}(x,\hat{x},\mathcal{V},\mathcal{V}^{*},A) & = & \int_{\mathbb{R}}d\tau\,\left[\frac{1}{2}\frac{d\mathcal{\mathfrak{M}}}{d\mu}(x,\tau)\,\mathcal{V}_{\alpha}^{*}(x)\mathcal{V}^{\alpha}(x)-\frac{d\mathcal{\mathfrak{E}}}{d\mu}(x,\tau)\,A_{\alpha}(x)\mathrm{Re}\left\{ \mathcal{V}^{\alpha}(x)\right\} \right]\nonumber \\
 &  & +\frac{1}{4\mu_{0}c}[F_{\alpha\beta}(x)+\delta f_{\alpha\beta}(x)]\cdot[F^{\alpha\beta}(x)+\delta f^{\alpha\beta}(x)]
\end{eqnarray}
such that $\mathfrak{S}[\hat{x},\mathcal{V},\mathcal{V}^{*},A]=\int_{\mathbb{A}^{4}(\mathbb{V_{\mathrm{M}}^{\mathrm{4}}},g)}d\mu(x)\,\mathfrak{L}(x,\hat{x},\mathcal{V},\mathcal{V}^{*},A)$

\subsection{Dynamics of a scalar electron}

The Lagrangian of a stochastic scalar electron with its interaction
is 
\begin{eqnarray}
L_{\mathrm{particle}}[\hat{x},\mathcal{V},\mathcal{V}^{*}] & \coloneqq & \int_{\mathbb{A}^{4}(\mathbb{V_{\mathrm{M}}^{\mathrm{4}}},g)}d\mathcal{\mathfrak{M}}(x,\tau)\,\frac{1}{2}\mathcal{V}_{\alpha}^{*}(x)\mathcal{V}^{\alpha}(x)\nonumber \\
 &  & -\int_{\mathbb{A}^{4}(\mathbb{V_{\mathrm{M}}^{\mathrm{4}}},g)}d\mathcal{\mathfrak{E}}(x,\tau)\,A_{\alpha}(x)\mathrm{Re}\left\{ \mathcal{V}^{\alpha}(x)\right\} \,.
\end{eqnarray}
Substituting this for (\ref{eq: EL eq}), 
\begin{equation}
\mathrm{Re}\left\{ m_{0}\mathfrak{D_{\tau}}\mathcal{V}^{\mu}(\hat{x}(\tau,\omega))+e\mathcal{\hat{V}}_{\nu}(\hat{x}(\tau,\omega))F^{\mu\nu}(\hat{x}(\tau,\omega))\right\} =0\,.\label{eq: Re=00007BEOM of S-particle=00007D}
\end{equation}
Where, the following are introduced with the Lorenz gauge $\partial_{\mu}A^{\mu}=0$
\cite{Nottale(2011)}:
\begin{equation}
\hat{\mathcal{V}}^{\mu}(x)\coloneqq\mathcal{V}^{\mu}(x)+\frac{i\lambda^{2}}{2}\partial^{\mu}\label{eq: OP of complex V}
\end{equation}
\begin{equation}
\mathfrak{D_{\tau}}=\hat{\mathcal{V}}^{\mu}(x)\partial_{\mu}\label{eq: complex D_tau}
\end{equation}
\begin{equation}
\mathfrak{D_{\tau}}A_{\mu}(\hat{x}(\tau,\omega))=\hat{\mathcal{V}}^{\nu}(\hat{x}(\tau,\omega))\partial_{\nu}A_{\mu}(\hat{x}(\tau,\omega))
\end{equation}
\begin{equation}
\mathfrak{D}_{\tau}^{*}A_{\mu}(\hat{x}(\tau,\omega))=\hat{\mathcal{V}}^{*\nu}(\hat{x}(\tau,\omega))\partial_{\nu}A_{\mu}(\hat{x}(\tau,\omega))
\end{equation}
 
\begin{thm}[Equation of stochastic motion]
\begin{leftbar}\label{EOM of a particle}The equation of ``stochastic''
motion of a scalar electron interacting with a field is
\begin{equation}
\boxed{d\mathcal{\mathfrak{M}}(x,\tau)\,\mathfrak{D_{\tau}}\mathcal{V}^{\mu}(x)=-d\mathcal{\mathfrak{E}}(x,\tau)\,\mathcal{\hat{V}}_{\nu}(x)F^{\mu\nu}(x)}\label{eq: EOM of S-particle}
\end{equation}
or
\begin{equation}
\boxed{m_{0}\mathfrak{D_{\tau}}\mathcal{V}^{\mu}(\hat{x}(\tau,\omega))=-e\mathcal{\hat{V}}_{\nu}(\hat{x}(\tau,\omega))F^{\mu\nu}(\hat{x}(\tau,\omega))}\,.\label{eq: EOM of S-particle-2}
\end{equation}
derived from the action integral (\ref{eq:Lagrangian-1}-\ref{eq:Lagrangian-1'}).
This is equivalent to the Klein-Gordon equation.\end{leftbar}\end{thm}
\begin{proof}
Let an arbitrary smooth $C^{1,0}$-function $f:\mathbb{A}^{4}(\mathbb{V_{\mathrm{M}}^{\mathrm{4}}},g)\times\mathbb{R}\rightarrow\mathbb{R}$
be a degree of freedom of the imaginary part of (\ref{eq: Re=00007BEOM of S-particle=00007D}),
namely, 
\begin{equation}
m_{0}\mathfrak{D_{\tau}}\mathcal{V}^{\mu}(\hat{x}(\tau,\omega))=-e\mathcal{\hat{V}}_{\nu}(\hat{x}(\tau,\omega))F^{\mu\nu}(\hat{x}(\tau,\omega))+\frac{i}{2m_{0}}\partial^{\mu}f(\hat{x}(\tau,\omega),\tau)\,.\label{eq: EOM of S-particle-1}
\end{equation}
Transforming $\mathfrak{D_{\tau}}\mathcal{V}^{\mu}+e/m_{0}\times\mathcal{\hat{V}}_{\nu}F^{\mu\nu}$
by employing (\ref{eq: complex D_tau}),
\begin{equation}
\mathfrak{D_{\tau}}\mathcal{V}^{\mu}+\frac{e}{m_{0}}\mathcal{\hat{V}}_{\nu}F^{\mu\nu}=\mathcal{\hat{V}}_{\nu}\left[\partial^{\nu}\mathcal{V}^{\mu}+\frac{e}{m_{0}}F^{\mu\nu}\right]=\mathcal{\hat{V}}_{\nu}\partial^{\mu}\mathcal{V}^{\nu}\,,
\end{equation}
since the identity 
\begin{equation}
\partial^{\alpha}\mathcal{V}^{\beta}-\partial^{\beta}\mathcal{V}^{\alpha}=\frac{e}{m_{0}}F^{\alpha\beta}
\end{equation}
derived from (\ref{eq: Complex velocity}) (also see Ref.\cite{Nottale(2011)}).
By substituting (\ref{eq: Complex velocity}) and the above for (\ref{eq: EOM of S-particle-1}),
\begin{eqnarray}
\mathcal{\hat{V}}_{\nu}\partial^{\mu}\mathcal{V}^{\nu}-\frac{i}{2m_{0}^{2}}\partial^{\mu}f & = & \left[i\lambda^{2}\times\partial_{\nu}\ln\phi+\frac{e}{m_{0}}A_{\nu}+\frac{i\lambda^{2}}{2}\partial_{\nu}\right]\times\partial^{\mu}\left[i\lambda^{2}\times\partial{}^{\nu}\ln\phi+\frac{e}{m_{0}}A{}^{\nu}\right]-\frac{i}{2m_{0}^{2}}\partial^{\mu}f\nonumber \\
 & = & \begin{gathered}\frac{1}{2}\partial^{\mu}\left[\frac{(i\hbar\partial_{\nu}+eA_{\nu})(i\hbar\partial^{\nu}+eA{}^{\nu})\phi-if\phi}{m_{0}^{2}\phi}\right]\end{gathered}
=0\,.
\end{eqnarray}
Thus, the Klein-Gordon equation is found by putting an arbitrary
constant $c^{2}$,
\begin{equation}
(i\hbar\partial_{\nu}+eA_{\nu})(i\hbar\partial^{\nu}+eA{}^{\nu})\phi-(m_{0}^{2}c^{2}+if)\phi=0\,.
\end{equation}
Thus, the imaginary force of $i/2m_{0}\times\partial^{\mu}f$ implies
a non-electromagnetic interaction. This interaction should be removed
or be rounded into the free-propagation term of a scalar electron
as the mass, hence, $f\equiv0$ is feasible in physics. Then, the
normal Klein-Gordon equation is derived;
\begin{equation}
(i\hbar\partial_{\nu}+eA_{\nu})(i\hbar\partial^{\nu}+eA{}^{\nu})\phi-m_{0}^{2}c^{2}\phi=0\,.
\end{equation}
Therefore, the equation of motion (\ref{eq: EOM of S-particle} or
\ref{eq: EOM of S-particle-2}) is equivalent to the Klein-Gordon
equation.

\end{proof}
Concerning {\bf Theorem \ref{Gauge_inv_V}}, the following theorem
is essential and obviously fulfilled since $\mathcal{V}$, $\mathcal{\hat{V}}$,
$\mathfrak{D_{\tau}}$ and $F$ are $U(1)$-gauge invariant which
the Klein-Gordon equation satisfies, too.
\begin{thm}[Gauge symmetry]
\begin{leftbar}The equation of motion (\ref{eq: EOM of S-particle}
or \ref{eq: EOM of S-particle-2}) satisfies the $U(1)$-gauge symmetry.\end{leftbar}
\end{thm}
Equations (\ref{eq: EOM of S-particle}) or (\ref{eq: EOM of S-particle-2})
are very similar style to classical dynamics, namely, $m_{0}dv^{\mu}/d\tau=-ev_{\nu}F^{\mu\nu}$.
Ehrenfest's theorem \cite{Ehrenfest} of it implies the average behavior
of this stochastic scalar electron like this classical form. It is
discussed at {\bf Section \ref{Sect4}}.

\subsection{Dynamics of fields}

Let us proceed the dynamics of the field radiated from a stochastic
scalar electron. The Maxwell equation is derived by the variation
of (\ref{eq:Lagrangian-1'}) with respect to $A\in\mathbb{V_{\mathrm{M}}^{\mathrm{4}}}$,
namely, $\partial_{\mu}[\partial\mathfrak{L}_{\mathrm{field}}/\partial(\partial_{\mu}A_{\nu})]-\partial\mathfrak{L}_{\mathrm{field}}/\partial A_{\nu}=0$.
Where, the Lagrangian density for a field is,
\begin{eqnarray}
\mathfrak{L}_{\mathrm{field}}[\hat{x},A] & = & -\int_{\mathbb{R}}d\tau\,\frac{d\mathcal{\mathfrak{E}}}{d\mu}(x,\tau)\,A_{\alpha}(x)\mathrm{Re}\left\{ \mathcal{V}^{\alpha}(x)\right\} \nonumber \\
 &  & +\frac{1}{4\mu_{0}c}[F_{\alpha\beta}(x)+\delta f_{\alpha\beta}(x)]\cdot[F^{\alpha\beta}(x)+\delta f^{\alpha\beta}(x)]\,.
\end{eqnarray}

\begin{thm}[Maxwell equation]
\begin{leftbar}\label{Maxwell}Let a D-progressive $\hat{x}(\circ,\bullet)$
be the trajectory of a stochastic scalar electron. The variation of
(\ref{eq:Lagrangian-1'}) with respect to a field $A\in\mathbb{V_{\mathrm{M}}^{\mathrm{4}}}$
derives the following Maxwell equation: 
\begin{equation}
\boxed{\partial_{\mu}[F^{\mu\nu}(x)+\delta f^{\mu\nu}(x)]=\mu_{0}\times j_{\mathrm{stochastic}}^{\mu}(x)}\label{eq: Maxwell eq}
\end{equation}
Where, $\delta f\in\mathbb{V_{\mathrm{M}}^{\mathrm{4}}}\otimes\mathbb{V_{\mathrm{M}}^{\mathrm{4}}}$
is a given field. The current of a stochastic scalar electron
\begin{equation}
j_{\mathrm{stochastic}}^{\mu}(x)\coloneqq\mathbb{E}\left\llbracket -ec\int_{\mathbb{R}}d\tau\,\mathrm{Re}\left\{ \mathcal{V}^{\mu}(x)\right\} \delta^{4}(x-\hat{x}(\tau,\bullet))\right\rrbracket \label{eq:j-stochastic}
\end{equation}
is equivalent to the current of Klein-Gordon particle 
\begin{equation}
j_{\mathrm{\mathrm{K\mathchar`-G}}}^{\mu}(x)=-\frac{iec\lambda^{2}}{2}\times g^{\mu\nu}\left[\phi^{*}(x)\mathfrak{D}_{\nu}\phi(x)-\phi(x)\mathfrak{D}_{\nu}^{*}\phi^{*}(x)\right]\,.\label{eq:j-KG}
\end{equation}
Where, $\mathfrak{D}^{\mu}\coloneqq\partial{}^{\mu}-ie/\hbar\times A{}^{\mu}(x)$.\end{leftbar}

\end{thm}
\begin{rem}
The tensor $\delta f^{\mu\nu}(x)$ is introduced to remove the field
singularity at the point of an electron. For the discussion of radiation
reaction in {\bf Volume II} \cite{Vol II}, the generated field has
to be separated into the homogeneous solution $F\in\mathbb{V_{\mathrm{M}}^{\mathrm{4}}}\otimes\mathbb{V_{\mathrm{M}}^{\mathrm{4}}}$
such that $\partial_{\mu}F^{\mu\nu}=0$, and the singularity as a
Coulomb field by $\partial_{\mu}\delta f^{\mu\nu}=\mu_{0}\times j_{\mathrm{stochastic}}^{\mu}$.
We regards (\ref{eq: Maxwell eq}) as the superposition of these two
equations. The detail of this discussion is in {\bf Volume II} \cite{Vol II}.\end{rem}
\begin{proof}
The derivation of the Maxwell equation (\ref{eq: Maxwell eq}) is
obvious. The current $j_{\mathrm{stochastic}}^{\mu}(x)$ is calculated
by using (\ref{eq: Complex velocity}):
\begin{eqnarray}
j_{\mathrm{stochastic}}^{\mu}(x) & = & -ec\int_{\mathbb{R}}d\tau\,\mathrm{Re}\left\{ \mathcal{V}^{\alpha}(x)\right\} p(x,\tau)\nonumber \\
 & = & \begin{gathered}\frac{\int_{\mathbb{R}}d\tau\,p(x,\tau)}{\phi^{*}(x)\phi(x)}\times j_{\mathrm{K-G}}^{\mu}(x)\end{gathered}
,\,\,x\in\bigcup_{\tau\in\mathbb{R}}\hat{x}(\tau,\varOmega)
\end{eqnarray}
Hereby, $j_{\mathrm{stochastic}}(x)\in\mathbb{V_{\mathrm{M}}^{\mathrm{4}}}$
satisfies $\partial_{\mu}j_{\mathrm{stochastic}}^{\mu}(x)=-ec\partial_{\mu}[\mathrm{Re}\left\{ \mathcal{V}^{\mu}(x)\right\} \int_{\mathbb{R}}d\tau\,p(x,\tau)]=0$
due to (\ref{eq: eq of continuity}) with its boundary condition $p(x,\tau=\partial\mathbb{R})=0$.
Of cause, $\partial_{\mu}j_{\mathrm{K\mathchar`-G}}^{\mu}(x)=0$ is
held, too. Thus,
\begin{equation}
\frac{\int_{\mathbb{R}}d\tau\,p(x,\tau)}{\phi^{*}(x)\phi(x)}=\mathrm{Constant}
\end{equation}
has to be imposed, $\partial_{\mu}[F^{\mu\nu}(x)+\delta f^{\mu\nu}(x)]=\mu_{0}j_{\mathrm{K\mathchar`-G}}^{\nu}$
is realized by $\int_{\mathbb{R}}d\tau\,p(x,\tau)/\phi^{*}(x)\phi(x)=1$.
\end{proof}

\begin{prop}
\begin{leftbar}\label{p-density}For the realization of the Klein-Gordon
equation and the Maxwell equation from the action integral (\ref{eq:Lagrangian-1}-\ref{eq:Lagrangian-1'}),
the following is required with respect to $x\in\bigcup_{\tau\in\mathbb{R}}\hat{x}(\tau,\varOmega)$:
\begin{eqnarray}
\phi^{*}(x)\phi(x) & \coloneqq & \int_{\mathbb{R}}d\tau\,\mathbb{E}\left\llbracket \delta^{4}(x-\hat{x}(\tau,\bullet))\right\rrbracket \nonumber \\
 & = & \int_{\mathbb{R}}d\tau\,p(x,\tau)
\end{eqnarray}
\end{leftbar}
\end{prop}
The plot of $\int_{\mathbb{R}}d\tau\,\mathbb{E}\llbracket\delta^{4}(x-\hat{x}(\tau,\bullet))\rrbracket$
denotes the distribution of a scalar electron in  $\hat{x}(\tau,\varOmega)\subset\mathbb{A}^{4}(\mathbb{V_{\mathrm{M}}^{\mathrm{4}}},g)$.

\section{Conclusion and discussion\label{Sect4}}

In this {\bf Volume I}, we discussed the relativistic and stochastic
kinematics of a scalar electron with its dynamics and a field. For
the kinematics of a particle, the D-progressive $\hat{x}(\circ,\bullet)$
on the Minkowski spacetime $(\mathbb{A}^{4}(\mathbb{V_{\mathrm{M}}^{\mathrm{4}}},g),\mathscr{B}(\mathbb{A}^{4}(\mathbb{V_{\mathrm{M}}^{\mathrm{4}}},g)),\mu)$
was defined as the extension from Nelson's (S3)-process \cite{Nelson(2001_book)}
at {\bf Definition \ref{D-progressive}} in {\bf Section \ref{Sect2}}.
It imposed the two types of the velocities $\mathcal{V}_{\pm}^{\mu}(\hat{x}(\circ,\bullet))$.
We needed to consider the probability density $p:\,\mathbb{A}^{4}(\mathbb{V_{\mathrm{M}}^{\mathrm{4}}},g)\times\mathbb{R}\rightarrow[0,\infty)$
characterized by {\bf Theorem \ref{FP eq}} and {\bf Proposition \ref{p-density}}.
The definition of the proper time $\ensuremath{d\tau}$ (\ref{eq: proper time original})
corresponding to one in classical dynamics is discussed, too. The
complex differential $\hat{d}$ (\ref{eq:2-3}) and the complex velocity
$\hat{\mathcal{V}}$ (\ref{eq: Complex velocity}) \cite{Nottale(2011)}
which are the main casts of the present model were also introduced.
In {\bf Section \ref{Sect3}}, the dynamics of a stochastic particle
was proposed. We introduced the new action integral (\ref{eq:Lagrangian-1}-\ref{eq:Lagrangian-1'})
corresponding to the form in classical dynamics. Hence, we could obtain
the dynamics of a stochastic particle and fields by the variational
calculus of this action integral. The restriction of an external field
is only Lorenz's gauge of $\partial_{\mu}A^{\mu}=0$.
\begin{conclusion}[System of a scalar electron and a field]
\begin{leftbar}\label{conclusion_1}Consider $(\mathbb{A}^{4}(\mathbb{V_{\mathrm{M}}^{\mathrm{4}}},g),\mathscr{B}(\mathbb{A}^{4}(\mathbb{V_{\mathrm{M}}^{\mathrm{4}}},g)),\mu)$
and $(\mathit{\Omega},\mathcal{F},\mathscr{P})$. A D-progressive
$\hat{x}(\circ,\bullet)\coloneqq\{\hat{x}(\tau,\omega)\in\mathbb{A}^{4}(\mathbb{V_{\mathrm{M}}^{\mathrm{4}}},g)|\tau\in\mathbb{R},\omega\in\varOmega\}$
characterized by \begin{snugshade}
\begin{equation}
\ensuremath{d\hat{x}^{\mu}(\tau,\omega)=\mathcal{V}_{\pm}^{\mu}(\hat{x}(\tau,\omega))d\tau+\lambda\times dW_{\pm}^{\mu}(\tau,\omega)}
\end{equation}
\end{snugshade}

\noindent is defined as the kinematics of a stochastic scalar electron
{\bf [Definition \ref{D-progressive}]}. The following action integral
{\bf [Definition \ref{Action integral}]}
\begin{eqnarray}
\mathfrak{S}[\hat{x},\mathcal{V},\mathcal{V}^{*},A] & = & \int_{\mathbb{R}}d\tau\int_{\mathbb{A}^{4}(\mathbb{V_{\mathrm{M}}^{\mathrm{4}}},g)}d\mathcal{\mathfrak{M}}(x,\tau)\,\frac{1}{2}\mathcal{V}_{\alpha}^{*}(x)\mathcal{V}^{\alpha}(x)\nonumber \\
 &  & -\int_{\mathbb{R}}d\tau\int_{\mathbb{A}^{4}(\mathbb{V_{\mathrm{M}}^{\mathrm{4}}},g)}d\mathcal{\mathfrak{E}}(x,\tau)\,A_{\alpha}(x)\mathrm{Re}\left\{ \mathcal{V}^{\alpha}(x)\right\} \nonumber \\
 &  & +\int_{\mathbb{A}^{4}(\mathbb{V_{\mathrm{M}}^{\mathrm{4}}},g)}d\mu(x)\,\frac{1}{4\mu_{0}c}[F_{\alpha\beta}(x)+\delta f_{\alpha\beta}(x)]\cdot[F^{\alpha\beta}(x)+\delta f^{\alpha\beta}(x)]
\end{eqnarray}
provides the following dynamics of a stochastic scalar electron {\bf [Theorem \ref{EOM of a particle}]}
and a field {\bf [Theorem \ref{Maxwell}]} characterized by $\mathscr{\mathcal{V}}\coloneqq(1-i)/2\times\mathcal{V}_{+}+(1+i)/2\times\mathcal{V}_{-}\in\mathbb{V_{\mathrm{M}}^{\mathrm{4}}}\oplus i\mathbb{V_{\mathrm{M}}^{\mathrm{4}}}$
and $F\in\mathbb{V_{\mathrm{M}}^{\mathrm{4}}}\otimes\mathbb{V_{\mathrm{M}}^{\mathrm{4}}}$:\begin{snugshade}
\begin{equation}
m_{0}\mathfrak{D_{\tau}}\mathcal{V}^{\mu}(\hat{x}(\tau,\omega))=-e\mathcal{\hat{V}}_{\nu}(\hat{x}(\tau,\omega))F^{\mu\nu}(\hat{x}(\tau,\omega))\label{eq:sug-KG eq}
\end{equation}
\begin{equation}
\partial_{\mu}\left[F^{\mu\nu}(x)+\delta f^{\mu\nu}(x)\right]=\mu_{0}\times\mathbb{E}\left\llbracket -ec\int_{\mathbb{R}}d\tau'\,\mathrm{Re}\left\{ \mathcal{V}^{\nu}(x)\right\} \delta^{4}(x-\hat{x}(\tau',\omega))\right\rrbracket \label{eq:sug-Maxwell}
\end{equation}
\end{snugshade}Here, the dynamics of (\ref{eq:sug-KG eq}) is equivalent
to the Klein-Gordon equation. These dynamics fulfill the $U(1)$-gauge
symmetry such that 
\begin{equation}
\begin{array}{ccc}
\phi'(x)=e^{-ie\varLambda(x)/\hbar}\times\phi(x)\,, &  & A'{}^{\alpha}(x)=A{}^{\alpha}(x)-\partial^{\alpha}\varLambda(x)\end{array}\,.
\end{equation}
\end{leftbar}
\end{conclusion}
We could hereby conclude {\bf Conjecture \ref{conj_kinematics}} is
demonstrated, however, how does this equation correspond to classical
behavior? It can be described by Ehrenfest's theorem and it is one
of the key idea for the investigation of radiation reaction in {\bf Volume II}.
\begin{thm}[Ehrenfest's theorem]
\begin{leftbar}\label{Ehrenfest}The expectation of (\ref{eq:sug-KG eq})
derives Ehrenfest's theorem of the Klein-Gordon equation. \end{leftbar}\end{thm}
\begin{proof}
\noindent Due to the identity $\mathbb{E}\llbracket dW_{\pm}^{\mu}(\tau,\bullet)\rrbracket=0$,
then, $\mathbb{E}\llbracket\mathcal{V}_{+}^{\mu}(\hat{x}(\tau,\bullet))\rrbracket=\mathbb{E}\llbracket\mathcal{V}_{-}^{\mu}(\hat{x}(\tau,\bullet))\rrbracket$
is satisfied. Considering the expectation of the equation of motion
(\ref{eq:sug-KG eq}),
\begin{eqnarray}
m_{0}\frac{d}{d\tau}\mathbb{E}\left\llbracket \mathcal{V}^{\mu}(\hat{x}(\tau,\bullet))\right\rrbracket  & = & m_{0}\frac{d}{d\tau}\mathbb{E}\left\llbracket \mathrm{Re}\{\mathcal{V}^{\mu}(\hat{x}(\tau,\bullet))\}\right\rrbracket \nonumber \\
 & \stackrel{(\ref{eq: Partial int formula})}{=} & \mathrm{Re}\left\{ \mathbb{E}\left\llbracket m_{0}\mathfrak{D_{\tau}}\mathcal{V}^{\mu}(\hat{x}(\tau,\bullet))\right\rrbracket \right\} \nonumber \\
 & = & \mathbb{E}\left\llbracket \mathrm{Re}\left\{ f^{\mu}(\hat{x}(\tau,\bullet))\right\} \right\rrbracket \,.
\end{eqnarray}
Where, $f^{\mu}(\hat{x}(\tau,\omega))\coloneqq-e\mathcal{\hat{V}}_{\nu}(\hat{x}(\tau,\omega))F^{\mu\nu}(\hat{x}(\tau,\omega))\in\mathbb{V_{\mathrm{M}}^{\mathrm{4}}}\oplus i\mathbb{V_{\mathrm{M}}^{\mathrm{4}}}$.
Since $d/d\tau\mathbb{E}\left\llbracket \hat{x}^{\mu}(\tau,\bullet)\right\rrbracket =\mathbb{E}\left\llbracket \mathcal{V}^{\mu}(\hat{x}(\tau,\bullet))\right\rrbracket $,\begin{snugshade}
\begin{equation}
m_{0}\frac{d^{2}}{d\tau^{2}}\mathbb{E}\left\llbracket \hat{x}^{\mu}(\tau,\bullet)\right\rrbracket =\mathbb{E}\left\llbracket \mathrm{Re}\left\{ f^{\mu}(\hat{x}(\tau,\bullet))\right\} \right\rrbracket \label{eq:Ehrenfest}
\end{equation}
\end{snugshade}

\noindent is derived and it is Ehrenfest's theorem which is an averaged
motion of a scalar electron.
\end{proof}
Due to {\bf Theorem \ref{Ehrenfest}}, the correspondence of the velocities
between classical and quantum dynamics is\begin{snugshade}
\begin{equation}
\begin{gathered}v^{\mu}(\tau)\leftrightarrow\frac{d}{d\tau}\mathbb{E}\left\llbracket \hat{x}(\tau,\bullet)\right\rrbracket =\mathbb{E}\left\llbracket \mathrm{Re}\{\mathcal{V}(\hat{x}(\tau,\bullet))\}\right\rrbracket \end{gathered}
\,.
\end{equation}
\end{snugshade}

\noindent Furthermore, $d/d\tau(v_{\mu}v^{\mu})=2\times v_{\mu}dv^{\mu}/d\tau=0$
has to be satisfied in classical dynamics. The present dynamics of
a stochastic particle provides the similar relation, too.
\begin{eqnarray}
\frac{d}{d\tau}\mathbb{E}\left\llbracket \mathcal{V}_{\mu}^{*}(\hat{x}(\tau,\bullet))\mathcal{V}^{\mu}(\hat{x}(\tau,\bullet))\right\rrbracket  & = & \mathbb{E}\left\llbracket \begin{gathered}\mathcal{V}_{\mu}^{*}(\hat{x}(\tau,\bullet))\cdot\mathfrak{D_{\tau}}\mathcal{V}^{\mu}(\hat{x}(\tau,\bullet))\\
+\mathfrak{D_{\tau}}^{*}\mathcal{V}_{\mu}^{*}(\hat{x}(\tau,\bullet))\cdot\mathcal{V}^{\mu}(\hat{x}(\tau,\bullet))
\end{gathered}
\right\rrbracket \nonumber \\
 & = & -\frac{\lambda^{2}e}{m_{0}}\times\mathbb{E}\left\llbracket \begin{gathered}\mathrm{Im}\{\mathcal{V}_{\mu}(\hat{x}(\tau,\bullet))\}\cdot\partial_{\nu}F^{\mu\nu}(\hat{x}(\tau,\bullet))\end{gathered}
\right\rrbracket \nonumber \\
 & = & -\frac{\lambda^{4}e}{2m_{0}}\times\int_{\mathbb{A}^{4}(\mathbb{V_{\mathrm{M}}^{\mathrm{4}}},g)}d\mu(x)\,\partial_{\mu}p(x,\tau)\cdot\partial_{\nu}F^{\mu\nu}(x)\nonumber \\
 & = & \frac{\lambda^{4}e}{2m_{0}}\times\int_{\mathbb{A}^{4}(\mathbb{V_{\mathrm{M}}^{\mathrm{4}}},g)}d\mu(x)\,p(x,\tau)\cdot\partial_{\mu}\partial_{\nu}F^{\mu\nu}(x)=0
\end{eqnarray}
Where, $p(x,\tau)|_{x\in\partial\mathbb{A}^{4}(\mathbb{V_{\mathrm{M}}^{\mathrm{4}}},g)}=0$
is selected. This result supports the Lorentz invariant $\mathbb{E}\llbracket\mathcal{V}_{\mu}^{*}(\hat{x}(\tau,\bullet))\mathcal{V}^{\mu}(\hat{x}(\tau,\bullet))\rrbracket=c^{2}$
(constant), too.
\begin{lem}
\begin{leftbar}The trajectory of a stochastic scalar electron satisfies
the following relation;
\begin{equation}
\frac{d}{d\tau}\mathbb{E}\left\llbracket \mathcal{V}_{\mu}^{*}(\hat{x}(\tau,\bullet))\mathcal{V}^{\mu}(\hat{x}(\tau,\bullet))\right\rrbracket =0\,.
\end{equation}
\end{leftbar}
\end{lem}
\noindent For $\varOmega_{\tau}^{\mathrm{ave}}\coloneqq\{\omega|\hat{x}(\tau,\omega)=\mathbb{E}\llbracket\hat{x}(\tau,\bullet)\rrbracket\}\subset\varOmega$,
we will quantize the LAD equation and demonstrate and the existence
of the following formula instead the radiation formula $dW_{\mathrm{QED}}/dt=q(\chi)\times dW_{\mathrm{classical}}/dt$
by using the present model in {\bf Volume II}:\begin{snugshade}
\begin{equation}
\frac{dW_{\mathrm{Stochastic}}}{dt}(\mathbb{\mathbb{E}}\llbracket\hat{x}(\tau,\bullet)\rrbracket)=\mathscr{P}(\varOmega_{\tau}^{\mathrm{ave}})\times\frac{dW_{\mathrm{classical}}}{dt}(\mathbb{E}\llbracket\hat{x}(\tau,\bullet)\rrbracket)
\end{equation}
\end{snugshade}

\section*{\addcontentsline{toc}{section}{Acknowledgement}Acknowledgment}

This work is supported by the Extreme Light Infrastructure Nuclear
Physics (ELI-NP) Phase II, a project co-financed by the Romanian Government
and the European Union through the European Regional Development Fund
- the Competitiveness Operational Programme (1/07.07.2016, COP, ID
1334).\\
$\,$

\end{document}